\documentclass[a4paper,USenglish]{lipics-v2018}
\usepackage{amsmath}
\usepackage{amsthm}
\usepackage{epsfig}
\usepackage{psfrag}
\usepackage{amssymb}
\usepackage{wrapfig}
\usepackage{paralist}
\usepackage{xcolor}
\usepackage{xspace}
\usepackage{nicefrac}
\usepackage{subcaption}
\usepackage{tabto}
\usepackage{dirtytalk}
\usepackage{pgfplots}
\usepackage{pgfgantt}
\usepackage{tikz}
\usetikzlibrary{arrows}
\usetikzlibrary{calc}
\usepackage{algpseudocode}
\usepackage{algorithmicx}
\usepackage[Algorithm]{algorithm}
\usepackage{subcaption}
\usepackage[T1]{fontenc}
\usepackage[utf8]{inputenc}

\makeatletter
\renewcommand\paragraph{\@startsection{myparagraph}{4}{0pt}%
                                    {-0.3\baselineskip}%
                                    {-1.5ex}%
                                    {\normalfont\normalsize\bfseries}}
\makeatother

\title{Combinatorial Algorithms for General Linear Arrow-Debreu Markets}

\author{Bhaskar Ray Chaudhury}{MPI for Informatics, Saarland Informatics Campus, Graduate School of Computer Science Saarbr\'ucken, Germany}{braycha@mpi-inf.mpg.de}{}{}
\author{Kurt Mehlhorn}{MPI for Informatics, Saarland Informatics Campus, Germany}{mehlhorn@mpi-inf.mpg.de}{}{}

\authorrunning{B.\,R. Chaudhury,  and K. Mehlhorn}
\Copyright{Bhaskar Ray Chaudhury, and Kurt Mehlhorn}


\subjclass{Theory of computation, Design and analysis of algorithms, Mathematical optimization}


\newtheorem{observation}[theorem]{Observation}

\newcommand{\Rnonneg}{\mathbb{R}_{\ge 0}}

\newtheorem{fact}[theorem]{Fact}

\newcommand{\block}{\mathit{block}}

\newcommand{\DM}{\mathrm{DM}}   

\newcommand{\norm}[1]{\lVert #1 \rVert}

\newcommand{\citeDM}{Duan-Mehlhorn:Arrow-Debreu-Market}
\newcommand{\set}[2]{\left\{  #1 \: \middle\vert \: #2 \right\}}
\newcommand{\xeq}{x_{\mathrm{eq}}}
\newcommand{\xmax}{x_{\max}}
\newcommand{\abs}[1]{\lvert #1 \rvert}
\nolinenumbers

\begin{document}

\maketitle

\begin{abstract}  We present a combinatorial algorithm for determining the market clearing prices of a general linear Arrow-Debreu market, where every agent can own multiple goods. The existing combinatorial algorithms for linear Arrow-Debreu markets consider the case where each agent can own all of one good only. We present an $\tilde{\mathcal{O}}((n+m)^7 \log^3(UW))$ algorithm where $n$, $m$, $U$ and $W$ refer to the number of agents, the number of goods, the maximal integral utility and the maximum quantity of any good in the market respectively. The algorithm refines the iterative algorithm of Duan, Garg and Mehlhorn using several new ideas. We also identify the hard instances for existing combinatorial algorithms for linear Arrow-Debreu markets. In particular we find instances where the ratio of the maximum to the minimum equilibrium price of a good is $U^{\Omega(n)}$ and the number of iterations required by the existing iterative combinatorial algorithms of Duan, and Mehlhorn and Duan, Garg, and Mehlhorn are high. Our instances also separate the two algorithms. \end{abstract}

\newpage
\pagenumbering{arabic}

\section{Introduction}

In a linear Arrow-Debreu market, there is a set $B$ of $n$ agents and a set $G$ of $m$ divisible goods. We will refer to an individual agent by $b_i$ for $i \in [n]$ and to an individual good by $g_j$ for $j \in [m]$. Each agent comes with a basket of goods to the market, more precisely, agent $b_i$ owns $w_{ij} \ge 0$ units of good $g_j$. The total supply of $g_j$ is then $\sum_i w_{ij}$ units. Moreover, the agents have utilities over the goods and $u_{ij} \ge 0$ is the utility derived by agent $b_i$ from one unit of good $g_j$. 

The goal is to find a positive price vector $p \in \Rnonneg^m$ and a non-zero flow  $f \in \Rnonneg^{n \times m}$ such that:\smallskip

\begin{compactenum}[\hspace{\parindent}a)]
\item For all $j \in [m]$: $\sum_i f_{ij} = \sum_i w_{ij} p_j$    \hfill (all goods are completely sold)
\item For all $i \in [n]$: $\sum_j w_{ij} p_j = \sum_j f_{ij}$ \hfill (agents spend all their income)
\item $f_{ij} > 0$ implies $u_{ij}/p_j = \max_{\ell \in [m]} u_{i\ell}/p_\ell$ \hfill (only bang-for-buck spending)
\end{compactenum}\smallskip

Such a price vector is called a vector of equilibrium prices. In a), the left hand side is the total flow (of money) into $g_j$ and the right hand side is the total value of all units of good $g_j$. In b), the left hand side is the income of agent $b_i$ under prices $p$ and the right hand side is his spending. In c), $\max_{\ell \in [m]} u_{i\ell}/p_\ell$ is the maximum ratio of utility to price (bang-for-buck) that agent $b_i$ can achieve. Agents spend only money on goods that give them the maximum bang-for-buck. For agent $b_i$ and good $g_j$, we use $x_{ij} = f_{ij}/ p_j$ for the amount of $g_j$ allocated to agent $b_i$. A price vector $p$ and an flow $f$ as above is called a \emph{market equilibrium}. 

We make the standard assumption that every agent likes at least one good, i.e., for all $i$, $\max_j u_{ij} > 0$, and that each good is liked by some agent, i.e., for all $j$, $\max_i u_{ij} > 0$. We also make the nonstandard assumption that for every proper subset $B' \subset B$ there is at least one good $g_j$ that is not completely owned by the agents in $B'$ and such that at least some $b_i \in B'$ is interested in $g_j$, i.e., $w_{kj} > 0$ for some $b_k \not\in B'$ and $u_{ij} > 0$ for some $b_i \in B'$.  In other words, there is no subset of agents that are only interested in the goods completely owned by them. References~\cite{Jain07,Devanur-Garg-Vegh,Duan-Mehlhorn:Arrow-Debreu-Market} show how to remove the nonstandard assumption. We assume that utilities $u_{ij}$ and weights $w_{ij}$ are integral and use $U = \max_{i \in [n], j\in [m]}u_{ij}$ to denote the maximum utility and $W = \max_{j \in [m]} \max_{i \in [n]}w_{ij}$ to denote the maximum weight. Then the budget available to any agent is bounded by $\sum_j w_{ij} p_j \le n W \max_j p_j$. 

Linear Exchange markets were introduced by Walras~\cite{Walras1874} back in 1874. Walras also argued that equilibrium prices exist. 
The first rigorous proof for the existence of equilibirum under strong assumptions was given by Wald~\cite{Wald36}. Arrow and Debreu~\cite{AD1954} gave the proof for the existence of equilibrium when the utility functions are concave. There has been substantial algorithmic research put into determining the equilibrium prices since the 60s. Codenotti et al.~\cite{CMV05} gives a surveys the algorithmic literature before 2004. While there are strongly polynomial approximation schemes for determining the equilibrium prices~\cite{JMS03, Garg-Kapoor, DV03}, the existence of a strongly polynomial exact algorithms still remains as an open question. There have been exact finite algorithms~\cite{Eaves76, Garg-Mehta-Sohoni-Vazirani}, exact weakly polynomial time algorithms~\cite{ Jain07,Ye2007, \citeDM, DGM:Arrow-Debreu} and the characterization of the equilibrium prices as a solution set of a convex program~\cite{Devanur-Garg-Vegh, Nenakov-Primak}.

A market equilibrium can be found in time polynomial in $n$, $m$, $\log U$ and $\log W$ by a number of different algorithms. Jain~\cite{Jain07} and Ye~\cite{Ye2007} gave algorithms based on the ellipsoid and the interior point method, respectively, and Duan and Mehlhorn~\cite{\citeDM} and Duan, Garg, and Mehlhorn~\cite{DGM:Arrow-Debreu} described combinatorial algorithms. The algorithm by Ye has a running time of $O(\max(n,m)^8 (\log UW)^2)$, see~\cite[footnote on page 2]{DGM:Arrow-Debreu}.

The combinatorial algorithms actually only solve a special case: $m = n$ and each agent is the sole owner of a good, i.e., $w_{ii} = 1$, and $w_{ij} = 0$ for $i \not= j$. The algorithm in~\cite{DGM:Arrow-Debreu} solves the special case in time  $O(n^7 \log^3(nU))$. A reduction for reducing the general case to the special case is known, see Section~\ref{new algorithm}. However, it turns a general problem with $n$ agents and $m$ goods into a special problem with $nm$ goods and hence leads to a running time of $\tilde{O}((nm)^7 \textup{poly}(\log(U)))$ (ignoring poly logarithmic dependencies on $n$). In an unpublished note, Darwish and Mehlhorn~\cite{Darwish-MehlhornAD,DarwishMasterThesis} have shown how to extend the algorithm in~\cite{\citeDM} to the general problem without going through the reduction. The resulting running time is $O(\max(n,m)^{10}\log^2(\max(n,m)UW))$. They were unable to generalize the approach in~\cite{DGM:Arrow-Debreu}.

\paragraph*{Our Contribution:} Our contribution is twofold: a combinatorial algorithm for the general problem and examples that are difficult for the algorithms in~\cite{\citeDM,DGM:Arrow-Debreu}.
We give a combinatorial algorithm for the general problem with running time $O((n+m)^7 $ $\log^3(nmUW))$. The algorithm refines the algorithm in~\cite{DGM:Arrow-Debreu} by several new ideas. We discuss them in Section~\ref{new algorithm}. In particular, in~\cite{DGM:Arrow-Debreu}, the number of iterations compared to~\cite{\citeDM} is reduced by a factor of $\Omega(n)$ by a modified price update rule . The modified update rule is subtle and heavily relies on the fact that there is a one to one correspondence between an agent and a good (one agent owns all of one good only) and this is not true in the general scenario and several of their crucial arguments break down. We come up with a novel price update rule that also highlights some new structure in the problem. 

We also give examples that are difficult for the algorithms in~\cite{\citeDM,DGM:Arrow-Debreu} amd where 
the equilibrium prices are exponential in $U$. Both algorithms are iterative and need $\mathcal{O}(n^5 \log(U))$ and $\mathcal{O}(n^4 \log(nU))$ iterations respectively. The examples force the algorithms into $\tilde{\Omega}(n^{4+\frac{1}{3}} \log(U))$ and $\tilde{\Omega}(n^{4} \log(U))$ iterations respectively (ignoring poly logarithmic dependencies on $n$). They separate the two algorithms.

\section{Determining the equilibrium price vector of the general linear Arrow-Debreu market}
\label{new algorithm}

For completeness we first show the well known reduction from the general case to the special problem in \cite{\citeDM,DGM:Arrow-Debreu} (where each agent owns all of one good only). This reduction is well-known. For each positive $w_{ij}$ we create an agent $b_{ij}$ and a good $g_{ij}$ owned by this agent. There is one unit of good $g_{ij}$. If $w_{ij} = 0$, there is no good $g_{ij}$ and no agent $b_{ij}$. We interpret $g_{i j}$ as the goods $g_j$ owned by $b_i$ and $b_{ij}$ as a copy of agent $b_i$. We define the utility derived by the agent $b_{ij}$ from one unit of good $g_{\ell k}$ as $\tilde{u}_{ij,\ell  k} = w_{\ell  k}\cdot u_{ik}$; here the factor $u_{ik}$ reflects that $b_{ij}$ is a copy of agent $b_i$ and $g_{\ell  k}$ is a copy of $g_k$ and the factor $w_{\ell  k}$ reflects that $b_\ell $ owns $w_{\ell  k}$ units of good $g_k$ but there is only copy of good $g_{\ell  k}$. 

\begin{lemma}
Let $p$ and $f$ be the market clearing price vector and the corresponding money flow for the above instance of the special case with $nm$ agents and goods. Then $\frac{p_{\ell k}}{w_{\ell k}}$ does not depend on $\ell $, but only on $k$. Let $\hat{p}_k = \frac{p_{\ell k}}{w_{\ell k}}$ and $\hat{f}_{i k} = \sum_{j \in [m]} \sum_{\ell  \in [n]} f_{ij,\ell k}$. Then $\hat{p}$ and $\hat{f}$ are the market clearing price vector and corresponding money flow for general case with utility matrix $u$ and weight matrix $w$.  
\end{lemma}

\begin{proof}
  Assume $\frac{w_{\ell k}}{p_{\ell k}} > \frac{w_{hk}}{p_{hk}} $. Let $b_{ij}$ be any arbitrary agent. Then $$\frac{\tilde{u}_{ij,\ell k}}{p_{\ell k}} = \frac{w_{\ell k} \cdot u_{ik}}{p_{\ell k}} > \frac{w_{h k} \cdot u_{ik}}{p_{hk}} = \frac{\tilde{u}_{ij,hk}}{p_{hk}}.$$ Hence agent $b_{ij}$ prefers good $g_{\ell k}$ over good $g_{hk}$. Since $b_{ij}$ is arbitrary,  $g_{hk}$ will not be sold at all, a contradiction. Thus $\frac{w_{\ell k}}{p_{\ell k}} = \frac{w_{hk}}{p_{hk}}$ for all $\ell $ and $h$. Now, it is easy to verify that every agent invests in goods that give him maximum utility to price ratio. Assume $\hat{f}_{i k}>0$. Then for any arbitrary $k'$ and $\ell '$, there is a $j$ and $\ell $, such that $$\frac{u_{ik}}{\hat{p}_k} = \frac{\tilde{u}_{ij,\ell k}}{w_{\ell k} \cdot \hat{p}_k} = \frac{\tilde{u}_{ij,\ell k}}{p_{\ell k}} \geq \frac{\tilde{u}_{ij,\ell 'k'}}{p_{\ell 'k'}} = \frac{\tilde{u}_{ij,\ell 'k'}}{w_{\ell 'k'} \cdot \hat{p}_{k'}} = \frac{u_{ik'}}{\hat{p}_{k'}}.$$ The equilibrium flow constraints are also easily verifiable,
   $$\sum_{i \in [n]} \hat{f}_{i k}=\sum_{i \in [n]} \sum_{j \in [m]} \sum_{\ell  \in [n]} f_{{ij},{\ell k}}=\sum_{\ell  \in [n]} p_{\ell k}=\sum_{\ell  \in [n]} w_{\ell k} \cdot \hat{p}_k$$
   $$\sum_{k \in [m]} \hat{f}_{i k}=\sum_{k \in [m]} \sum_{j \in [m]} \sum_{\ell  \in [n]} f_{{ij},{\ell k}}=\sum_{j \in [m]} p_{ij}=\sum_{j \in [m]} w_{ij} \cdot \hat{p}_j$$
\vspace{-1.2cm}\par \end{proof}
We now present our algorithm that does not rely on this reduction.
\subsection{The Algorithm}

Figure~\ref{Program} shows the algorithm. Similar to the algorithms in~\cite{DPSV08,\citeDM,DGM:Arrow-Debreu}, the algorithm is iterative and flow based. For the description of the algorithm, we need the concepts of an equality network, of a balanced flow, of the set of high-surplus buyers, and of the flow- and price-update. 

\begin{algorithm}[t]
\begin{algorithmic}[1]
 \State Set $p_i \gets 1 \quad \forall j \in [n]$.
 \State Set $\varepsilon \gets {1}/({8\cdot (n+m)^{4(n+m)}(UW)^{3(n+m)}})$.
 \While{$\norm{r_f}_2 > \varepsilon$ , where $f$ is a balanced flow in $N_p$}
      \State Let $S$ be the set of high-surplus agents w.r.t $f$ in $N_p$.
      \State $x \gets \min(\xeq,x_{23},x_{24},x_{13},x_2,x_{\max})$.
      \State \parbox{0.95\textwidth}{Multiply prices of goods in $\Gamma(S)$ by $x$ and update $f$, $p$ to $f'$ and $p'$ as in $(\ref{eq-1})$ - $(\ref{eq-4})$  and $N_p$ to $N_{p'}$.}
      \State Let $f''$ be the balanced flow in $N_{p'}$. 
      \State Set $p \gets p'$ and $f \gets f''$.
    \EndWhile
  \State Round $p$ to equilibrium prices.      
  
\end{algorithmic}
\caption{Combinatorial algorithm for determining the equilibirum prices in the general linear Arrow-Debreu market}\label{Program}
\end{algorithm}


\paragraph*{Equality Network $N_p$:} For a price vector $p$ the equality network $N_p$ is a flow network with vertices  $s \cup t \cup B \cup G$ and edges
    \begin{itemize}
        \item $(s,b_i)$ with capacity $\sum_{j \in [m]}w_{ij}\cdot p_j$ for all $i \in [n]$.
        \item $(g_j,t)$ with capacity $\sum_{i \in [n]}w_{ij}\cdot p_j$ for all $j \in [m]$.
        \item $(b_i,g_j)$ with capacity $\infty$ iff ${u_{ij}}/{p_j} \geq {u_{ik}}/{p_k}$ for all $k \in [m]$.
    \end{itemize}
For any $B' \subseteq B$ let $\Gamma(B') = 
\set{g_j}{(b_i,g_j) \in N_p \text{ for some } b_i \in B'}$ denote the neighborhood of $B'$ in $N_p$ (all the goods agents in $B'$ may invest on).
    
\paragraph*{Surpluses and Surplus vector $r_f$:}  Let $f$ be a valid flow in the equality network $N_p$, and let $f_{ij}$, $f_{si}$ and $f_{jt}$ denote the flow along the edge $(b_i,g_j)$, $(s,b_i)$ and $(g_j,t)$ respectively.  We define the surplus $r_f(b_i)$ of the agent $b_i$ and $r_f(g_j)$ of the good  $g_j$, as $r_f(b_i) = \sum_{j \in [m]}w_{ij}\cdot p_j - f_{si} $ and $r_f(g_j) = \sum_{i \in [n]}w_{ij}\cdot p_j - f_{jt}$ respectively. Surpluses are always non-negative. We define the surplus vector $r_f \in \mathbb{R}^{n}$ as $\langle r_f(b_1),r_f(b_2), \ldots r_f(b_n) \rangle$.
    
\paragraph*{Balanced Flow:} A balanced flow $f$ is a valid flow in $N_p$ with minimum norm $\norm{r_f}_2^2$. Every balanced flow is a maximum flow. Additionally, if $f$ is a balanced flow and $f_{ij}$ and $f_{i'j}$ are positive, then the surpluses of the agents $b_i$ and $b_{i'}$ are the same and if $(b_i,g_j)$ and $(b_{i'},g_j)$ are edges of $N_p$ and $r_f(b_i) > r_f(b_{i'})$ then $f_{i'j} = 0$. This essentially follows from the fact that the L2 norm of a vector reduces as the components move closer to each other in magnitude (while the L1 norm remain constant). The following algorithmic property of balanced flows will be useful.
    
    \begin{lemma}
     Balanced Flows can be computed with at most $n$ max flow computations~\cite{DPSV08} and by one parameterized flow computation~\cite{Darwish-Mehlhorn}. 
    \end{lemma}

 \paragraph*{High Surplus Agents $S$ and Goods in High Demand $\Gamma(S)$:} Let $f$ be a balanced flow and assume that there is surplus. The high surplus edges and high demand goods w.r.t.~$f$ in $N_p$ are defined exactly as  in~\cite{DGM:Arrow-Debreu}. Renumber the agents in order of decreasing surplus so that $b_1$ has the highest surplus and $b_n$ has the lowest surplus. Let $\ell$ be minimal such that $r_f(b_\ell) > r_f(b_{\ell + 1})$ and for every $k$ such that 
$r_f(b_\ell) > r_f(b_k) \ge r_f(b_\ell)/(1 + 1/n)$, $f_{sk} = 0$ and $w_{kj} = 0$ for every $g_j \in \Gamma(S)$, where $S = \{b_1,\ldots,b_\ell\}$. If no such $\ell$ exists, let $\ell = n$. We refer to $S$ as the set of high surplus agents and to $\Gamma(S)$ as the set of high demand goods. The surplus of the goods in $\Gamma(S)$ is zero since the agents in $S$ have positive surplus. There is no flow on edges $(b_i,g_j) \in N_p$ with $b_i \not\in S$ and $g_j \in \Gamma(S)$. 
       
The algorithm in~\cite{DGM:Arrow-Debreu} that determines $S$ in time $\mathcal{O}(n^2)$ can be easily generalized and we do not discuss it here.

\paragraph*{Price and Flow Update:}
Like the earlier combinatorial algorithms, our algorithm is a multiplicative price update algorithm. It works in phases and in each phase we compute the balanced flow $f$ in $N_p$ and determine the high surplus agents $S$ and high demand goods $\Gamma(S)$. We  increase the prices of the goods in $\Gamma(S)$ as well as the money flow into them by the same factor $x>1$ (so we change the flow $f$ in $N_p$ to $f'$ in $N_{p'}$). Formally,
\begin{align} \label{eq-1}
  f'_{ij} &= \left\{
        \begin{array}{ll}
            x \cdot f_{ij} & \quad g_j \in \Gamma(S) \\
            f_{ij}         & \quad g_j \notin \Gamma(S)\\
        \end{array}
    \right. \qquad 
 f'_{si} = \left\{
        \begin{array}{ll}
            x \cdot f_{si}   & \quad b_i \in S \\
            f_{si}         & \quad b_i \notin S\\
        \end{array}
    \right. \\
\label{eq-4}
   f'_{jt} &= \left\{
        \begin{array}{ll}
            x \cdot f_{jt} & \quad g_j \in \Gamma(S) \\
            f_{jt}         & \quad g_j \notin \Gamma(S)\\
        \end{array}
    \right. \qquad p'_i = \left\{
        \begin{array}{ll}
            x \cdot p_i      & \quad g_j \in \Gamma(S) \\
            p_i             & \quad g_j \notin \Gamma(S)\\
        \end{array}
    \right.
  \end{align} 



  

\paragraph*{The Factor $x$:} We define
$x = \min(\xeq,x_{23},x_{24},x_{13},x_2,x_{\max})$, where the quantities on the right are defined below. 

Since we increase the prices of all the goods in $\Gamma(S)$ by the same factor $x$, the only equality edges that may disappear are the ones that connect an agent from $B \setminus S$ to a good from $\Gamma(S)$. Since $f$ is a balanced flow and the agents in $S$ have strictly higher surplus than the ones in $B \setminus S$, the edges from agents in $B \setminus S$ to goods in $\Gamma(S)$ in $N_{p}$ carry no flow and hence them disappearing will not lead to a violation of the flow constraints.  The new edges that appear will connect an agent in $S$ to a good in $G \setminus \Gamma(S)$. Therefore we define, 
\begin{align*}
\xeq = \min \set{\frac{u_{ij}}{p_j} \cdot \frac{p_k}{u_{ik}}}{b_i \in S, (b_i,g_j) \in N_p, g_k \notin \Gamma(S)}
\end{align*}
This is the minimum $x$ at which a new equality edge appears in the network. 

Next we consider how the surpluses of the agents are affected. Observe that $r_{f'}(b_i) = r_f(b_i) + x \cdot (\sum_{g_j \in \Gamma(S)}(w_{ij}p_j - f_{ij}))$. Therefore all the surpluses vary linearly with $x$. We now, introduce 5 classes of agents similar to the ones in~\cite{DGM:Arrow-Debreu}. Figure~\ref{definition of types} illustrates this definition.

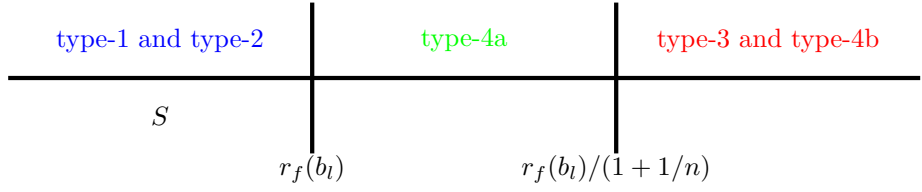
\begin{figure}[t]
\psfrag{bl}{$b_\ell$}
\begin{center} 
\begin{tikzpicture}

\draw[black, ultra thick] (-6,0) -- (6,0);
\draw[black, ultra thick] (-2,1) -- (-2,-1);
\draw[black, ultra thick] (2,1) -- (2,-1);

\node at (-4,0.5) {\textcolor{blue}{type-1 and type-2}};
\node at (-4,-0.5) {$S$};
\node at (-2,-1.2) {$r_f(b_l)$};
\node at (0,0.5) {\textcolor{green}{type-4a}};
\node at (2,-1.2) {$r_f(b_l)/(1+1/n)$};
\node at (4,0.5) {\textcolor{red}{type-3 and type-4b}};

\end{tikzpicture} 
\end{center}
\caption{\label{definition of types} The horizontal line indicates the agents ordered by decreasing surplus from left to right. $b_\ell$ is the agent of smallest surplus in $S$. All agents with surplus in $r_f(b_\ell), r_f(b_\ell)/(1 + 1/n)]$ are type-4a agents. They have no outflow and do not own goods in $\Gamma(S)$. }
\end{figure}

\begin{itemize}
\item Type-1 agent: An agent is a type-1 agent if it belongs to $S$ and its surplus increases by the price change. Formally, $b_i$ is type-1 if $b_i \in S$ and $\sum_{g_j \in \Gamma(S)} w_{ij} p_j > \sum_{g_j \in \Gamma(S)} f_{ij}$.

\item Type-2 agent: An agent is a type-2 agent if it belongs to $S$ and its surplus does not increase by the price change. Formally, $b_i$ is type-2 if $b_i \in S$ and $\sum_{g_j \in \Gamma(S)} w_{ij} p_j \leq \sum_{g_j \in \Gamma(S)} f_{ij}$.

\item Type-3 agent: An agent is type-3 agent if it does not belong to $S$ and its surplus increases by the price change, i.e., it partially owns a good in $\Gamma(S)$. Formally, $b_i \not\in S$ and $w_{ij} > 0$ for some $g_j \in \Gamma(S)$. 

\item Type-4a agent: An agent is a type-4a agent if it does not belong to $S$ and has no ownership in a good in $\Gamma(S)$ and its surplus is at least $r_{\min}/({1 + {1}/{n}})$, where $r_{\min}$ is the minimum surplus of any agent in $S$. Formally, $b_i \not\in S$, $w_{ij} = 0$ for all $j \in \Gamma(S)$ and $r_f(b_i) \ge r_{\min}/({1 + {1}/{n}})$. By definition of $S$ such $b_i$ have no outflow, i.e., $f_{si} = 0$, and no ownership of any good in $\Gamma(S)$, i.e., $w_{ij} = 0$ for any $j \in \Gamma(S)$. 

\item Type-4b agent: An agent is a type-4b agent if it does not belong to $S$ and has no ownership in a good in $\Gamma(S)$ and its surplus is strictly less than  $r_{\min}/(1 + {1}/{n})$, where $r_{\min}$ is the minimum surplus of any agent in $S$. Formally, $b_i \not\in S$, $w_{ij} = 0$ for all $j \in \Gamma(S)$ and $r_f(b_i) < r_{\min}/(1 + {1}/{n})$.
\end{itemize}
We abbreviate the above sets of agents of each type as $T_1$, $T_2$, $T_3$, $T_{4a}$ and $T_{4b}$. Notice that the surpluses of agents in $T_{4a} \cup T_{4b}$ remain unaffected by increase of price of goods in $\Gamma(S)$. We now define $x_{23}$, $x_{24}$, $x_{13}$ as the minimal $x$ such that $r_{f'}(b_i) = r_{f'}(b_j)$ for some $b_i \in T_2, b_j \in T_3$ and $b_i \in T_2, b_j \in T_{4b}$ and $b_i \in T_1,b_j \in T_3$ respectively. We also define $x_2$ as the minimal $x$ when $r_{f'}(b_i) = 0$ where $b_i \in T_2$. 

Finally, 
\[ x_{\max} =    \begin{cases} 
      1 + \frac{1}{Rn^3} & \text{when}  \min_{g_j \in \Gamma(S)} p_j < Rn^4 m W\\
      1 + \frac{1}{Rkn^2} & \text{when} \min_{g_j \in \Gamma(S)} p_j \geq Rn^4 m W,
\end{cases}\]
where $k$ is the number of agents in $S$ that partially own a good in $\Gamma(S)$ and $R = 8 e^2$. Note that $k$ is at least the number of type-1 agents.\footnote{In~\cite{DGM:Arrow-Debreu}, $\xmax$ is defined as 
\[ x_{max} = \begin{cases} 
      1+ \frac{1}{Rn^3} &  \text{if there are type-3 agents}\\
      1+ \frac{1}{Rkn^2} & \text{otherwise, where $k$ = number of type-1 agents}.
   \end{cases} \]
The revised definition we change the classification of iterations. We do not classify them based on the type of the agents in $S$, but by looking at the smallest price of a good in $\Gamma(S)$ and we have a larger $k$ ($k$ is at least the number of type-1 agents).}
We distinguish light and heavy iterations. An iteration is light, if $p_j < R n^4 m W$ for some good $g_j \in \Gamma(S)$ (some good in demand is not heavily priced), and heavy if $p_j \ge R n^4 m W$ for all goods $g_j \in \Gamma(S)$ (all goods in demand are heavily priced). 

\paragraph*{Effect of an Iteration:} We multiply the prices of the goods in $\Gamma(S)$ by $x$ and update $f$, $p$ to $f'$ and $p'$ as in (1)-(4)  and $N_p$ to $N_{p'}$. Note that the goods completely sold w.r.t.~$f$ in $N_p$ are also completely sold with respect to $N_{p'}$. Maximizing and balancing the flow will not increase the surplus of a good from zero to a positive value. Thus all goods that are completely w.r.t.~$f$ in $N_p$ are also completely sold w.r.t. the balanced flow $f''$ in $N_{p'}$. The 2-norm of $r_{f''}$ is at most the $L2$ norm of $r_{f'}$. 

We stop updating the prices once the $L2$ norm of the surplus vector is at most $\varepsilon = {1}/(8(n+m)^{4(n+m)}(U \cdot W)^{3(n+m)})$. 

\subsection{Analysis of the Algorithm}
Before we present the analysis of our algorithm we briefly indicate why the price-update scheme used in~\cite{DGM:Arrow-Debreu} does not generalize. The algorithm in~\cite{DGM:Arrow-Debreu} reduces the number of iterations of the algorithm in~\cite{\citeDM} by a multiplicative factor of $\Omega(n)$ by performing a more careful selection of the set $S$ and crucial change in the price update rule. In particular they set 
\[ x_{max} = \begin{cases} 
      1+ \frac{1}{Rn^3} &  \text{if there are type-3 agents}\\
      1+ \frac{1}{Rkn^2} & \text{otherwise, where $k$ = number of type-1 agents}.
   \end{cases} \]
Since in the special case, every agent owns only all of one good, $k$ equals the number of type-1 agents in $S$ and the number of goods in $\Gamma(S)$ in all iterations that do not involve a type-3 agent. We enlist two crucial arguments that are necessary (not sufficient) to reduce the number of iterations by $\Omega(n)$ from that in~\cite{\citeDM} and their dependencies on $k$ as follows,
\begin{enumerate}
    \item There is a decrease in the $L2$ norm of the surplus vector in every balancing iteration without a type-3 agent. This argument crucially relies on $k$ being equal to the number of type-1 agents.
    \item The total multiplicative increase of the $L2$ norm of the surplus vector in all $x_{\mathit{max}}$ iterations without a type-3 agent is at most $(nU)^{\mathcal{O}(n)}$. This claim relies on the fact that $k$ equals the number of goods in $\Gamma(S)$.
\end{enumerate}
In the general scenario we do not have such one to one correspondence between the agents and the goods. Therefore if we choose $k$ to be either the number of type-1 agents in $S$ or the number of goods in $\Gamma(S)$, then one of the claims from above will fail. Thus we need at least a different classification of the iterations or a different choice of $k$ than the ones used in~\cite{DGM:Arrow-Debreu}. In our algorithm we do both and thereby highlighten more hidden structure in the problem. 

Like the algorithms in~\cite{\citeDM,DGM:Arrow-Debreu}, in every iteration Algorithm \ref{Program} only increases the prices of goods that are completely sold (since $f$ is a balanced flow and the agents in $S$ have positive surplus, the goods in $\Gamma(S)$ will have zero surplus). Since the sum of surpluses of the agents equals the sum of surpluses of the goods, both are therefore non-decreasing during a price update. Also the goods completely sold w.r.t.~$f$ in $N_p$ also remain sold w.r.t.~$f''$ in $N_{p'}$ in every iteration. Therefore the sum of surpluses of the agents is non-increasing throughout the algorithm. Initially (when the price of every good $g_j$ is $1$) this sum is at most $nmW$. 

\begin{observation}
\label{boundedsurplusum}
 The $L1$ norm of the surplus vector is non-increasing throughout the algorithm and is at most $n m W$.
\end{observation}

\subsubsection{An Upper Bound on the Maximum Price}

Observe that in every iteration of the algorithm there is a good with unit price. This follows from the fact that the goods that are completely sold stay completely sold and since only the prices of goods that are completely sold are increased, the price of goods that are not completely sold is equal to the initial price and hence equal to one. We now derive an upper bound on the maximum price of a good.   
\begin{lemma}
\label{ratio-bounds} At any time during the course of the algorithm the maximum price is at most 
$\max(2,U)^{m-1} \cdot W^{2m-2}$ and the maximum budget is at most  $\max(2,U)^{m} \cdot W^{2m}$.
 \end{lemma}
 \begin{proof}
  Let us renumber the goods in  increasing order of their prices. We show that $p_i \leq \max(2,U)^{i-1} \cdot W^{2i-2}$ by induction on $i$. The smallest price is 1 and this establishes the induction base. Consider an arbitrary $i$ and let $\hat{G}$ denote the set of goods $\left\{g_{i+1},...,g_m\right\}$ and $\hat{B}$ be the agents investing on the goods in $\hat{G}$. We will derive a bound on $p_{i+1}$. We distinguish two cases. 

Assume first that there is an agent in $b_h \in \hat{B}$ that is also interested in a good in $G \setminus \hat{G}$; say $b_h$ invests on good $g_j \in \hat{G}$ and is interested in good $g_\ell \in G \setminus \hat{G}$. Then  ${u_{hj}}/{p_j} \geq {u_{h\ell}}/{p_\ell}$ and hence $p_{i+1} \leq p_j \leq U \cdot p_\ell \leq U \cdot p_i \leq U \cdot \max(2,U)^{i-1} \cdot W^{2i-2} \leq \max(2,U)^i \cdot W^{2i}$.

Assume next that the agents in $\hat{B}$ are only interested in the goods in $\hat{G}$. Then there must be a good in $\hat{G}$, say $g_k$, that is partially owned by an agent in $ B \setminus \hat{B}$. Otherwise, the agents in $\hat{B}$ will only be interested in goods completely owned by them. Therefore let $b_{h} \in B \setminus \hat{B}$ be the agent that partially owns $g_k$. The budget $m_h$ of $b_h$ is at least $w_{hk}p_k$ and hence  at least $p_k$. 
Since $b_{h}$ invests only in goods in $G \setminus \hat{G}$, its budget is at most the total value of the goods in $G \setminus \hat{G}$. Thus 
\[ m_{h} \leq W \cdot \sum_{j \in [i]} p_j \leq W \cdot \sum_{j \in [i]}  \left( \max(2,U)^{j-1} W^{2j-2}\right) \leq W  \max(2,U)^{i} W^{2i-1}=\max(2,U)^{i} W^{2i} \]
and hence $p_{i+1} \le p_k \le \max(2,U)^{i} W^{2i} $.
The maximum budget of an agent is at most $W$ times the total price of all goods and hence is bounded by $W \cdot \sum_{j \in [m]}\left( \max(2,U)^{j-1} \cdot W^{2j-2}\right) \leq \max(2,U)^{m} \cdot W^{2m}$.
\end{proof}

In particular, when $n = m$ and $w$ is an identity matrix, we have an upper bound of $\max(2,U)^{n-1}$ for the highest price of a good in contrast to $\max(n,U)^{n}$ in~\cite{\citeDM,DGM:Arrow-Debreu}. Also note that the maximum price of a good and the maximum budget of an agent is independent of the number of agents. We now separately bound the $\xmax$-iterations ($x = \xmax$) and balancing iterations ($x < \xmax$).

\subsubsection{Bounding the number of $\xmax$-iterations}

In this section we will be bounding the light and the heavy $\xmax$-iterations separately. For bounding both classes of iterations, we use upper bounds on the prices of the goods. A more aggressive price update scheme is used for the heavy $\xmax$-iterations as the prices of all goods in $\Gamma(S)$ in such iterations are high. Such aggressive price update may apparently result in a significant multiplicative increase in the $L2$ norm of the surplus vector. We address this concern in the next subsection. We first show that $\xmax$-iterations where there is at least one type-3 agent are light.
 \begin{lemma}
$\xmax$-iterations with at least one type-3 agent are light. 
 \end{lemma}
 \begin{proof}
Let $b_i$ be a type-3 agent. Then there is a $g_j \in \Gamma(S)$ with $w_{ij} > 0$. The additive increase in the surplus of a type-3 agent $b_i$ during an $\xmax$-iteration is at least ${w_{ij}\cdot p_j}/{Rn^3}$. Since the total surplus is always at most $n m W$ (by Observation \ref{boundedsurplusum}), the increase in the surplus of any agent is at most $n m W$. This immediately implies that $p_j \leq Rn^4 m W$.   
 \end{proof}
Now we bound the number of light $\xmax$-iteration.

\begin{lemma}
 The number of light $\xmax$-iterations is at most $10R^2 n^3 m \log(nmW)$.
\end{lemma}
\begin{proof}
Assume otherwise. Then there exists a good $g_j$ that is the minimum priced good in $\Gamma(S)$ in more than $10R^2 n^3\log(nmW)$ iterations. Before the last such iteration, $p_j > (1 + \frac{1}{Rn^3})^{10R^2 n^3\log(nmW)} > e^{5R\log(nmW)} = n^{5R} m^{5R}W^{5R} > Rn^4 m W$, which is a contradiction. The second to last inequality uses the fact that $1+x > e^{\frac{x}{2}}$ for $0 < x \leq 1$. 
\end{proof}

We turn to heavy $\xmax$-iterations. For such iterations, there exists no type-3 agent and hence the goods in $\Gamma(S)$ are completely owned by the agents in $T_1 \cup T_2$. Thus 
\[ \sum_{b_i \in T_1 \cup T_2} \sum_{g_j \in \Gamma(S)} f_{i j} = \sum_{g_j \in \Gamma(S)} \sum_{i \in T_1 \cup T_2} w_{ij} p_j = \sum_{b_i \in T_1}\sum_{g_j \in \Gamma(S)} w_{ij} \cdot p_j + \sum_{b_i \in T_2}\sum_{g_j \in \Gamma(S)} w_{ij} \cdot p_j,\]
and hence 
\begin{equation}\label{type1 versus type 2 surplus} \sum_{b_i \in T_1} \sum_{g_j \in \Gamma(S)} (w_{ij}\cdot p_j -f_{i j})= \sum_{b_i \in T_2}\sum_{g_j \in \Gamma(S)}(f_{i j}- w_{ij} \cdot p_j). \end{equation}
Note that type-1 and type-2 agents can own goods outside of $\Gamma(S)$. However the above relation will help us prove that this \say{excess budget} is at most $n m W$. In fact the following lemma plays a pivotal role to bound the multiplicative increase in the $L2$ norm of the surplus vector in the next subsubsection.

\begin{lemma}
\label{technical}
 For every agent $b_i \in S$ in a heavy $\xmax$-iteration, $\sum_{g_j \notin \Gamma(S)} w_{ij}p_j \leq n m W$.
\end{lemma}
\begin{proof}
For $b_i \in S$, we have $r_f(b_i) = \sum_{g_j \not\in \Gamma(S)} w_{ij}p_j + \sum_{g_j \in \Gamma(S)} (w_{ij} p_j - f_{i j})$. 
For $b_i \in T_1$, this implies $\sum_{g_j \not\in \Gamma(S)} w_{ij}p_j \le r_f(b_i) \le n m W$. For $b_i \in T_2$, using (\ref{type1 versus type 2 surplus}) 
\begin{align*}
\sum_{g_j \not\in \Gamma(S)} w_{ij}p_j &= r_f(b_i) + \sum_{g_j \in \Gamma(S)} (f_{i j} - w_{ij} p_j )
                                           \le r_f(b_i) + \sum_{b_h \in T_2} \sum_{g_j \in \Gamma(S)} (f_{hj} - w_{hj} p_j )\\
                                           &= r_f(b_i) + \sum_{b_h \in T_1} \sum_{g_j \in \Gamma(S)} (w_{hj}p_j  - f_{hj})    
                                           \le r_f(b_i) + \sum_{b_h \in T_1} r_f(b_h)
                                           \le n m W.
\end{align*} \vspace{-1.2cm}\par
\end{proof} 

In a heavy $\xmax$-iteration, the price of any good in $\Gamma(S)$ is at least $Rn^4mW$ and hence the budget of any agent that partially owns a good in $\Gamma(S)$ is at least that much. By the above, the ownership of the goods outside $\Gamma(S)$ contribute very little to the budget of such agents. 
Thus any multiplicative increment on the prices of the goods in $\Gamma(S)$ will inflict an almost equal multiplicative increase in the budget of such agents.

\begin{lemma}
 The number of heavy $\xmax$-iterations is $\mathcal{O}(n^3 m \cdot \log(WU))$.
\end{lemma}

\begin{proof}
Consider any heavy $\xmax$ iteration and let $m_i$ denote the budget of agent $b_i$. For any agent $b_i$ that partially owns a good in $\Gamma(S)$, $m_i \geq Rn^4 m W$. The budget $m_i$ of any agent $b_i$ that partially owns a good in $\Gamma(S)$, increases as follows (new budget denoted by $m_i'$):
 \begin{align*}
 m_i' &= \sum_{g_j \notin \Gamma(S)}w_{ij}p_j + (1 + \frac{1}{Rkn^2}) \cdot \sum_{g_j \in \Gamma(S)}w_{ij}p_j\\
 &\ge (1 + \frac{1}{Rkn^2}) \cdot (1+\frac{1}{n^3})^{-1} \cdot (1 + \frac{1}{n^3}) \cdot \sum_{g_j \in \Gamma(S)}w_{ij}p_j\\
&=(1 + \frac{1}{Rkn^2}) \cdot (1+\frac{1}{n^3})^{-1} \cdot   \left(\sum_{g_j \in \Gamma(S)}{w_{ij}p_j} + \frac{\sum_{g_j \in \Gamma(S)}w_{ij}p_j}{n^3}\right)\\
 &\geq (1 + \frac{1}{Rkn^2}) \cdot (1+\frac{1}{n^3})^{-1} \cdot   \left(\sum_{g_j \in \Gamma(S)}{w_{ij}p_j} + nmW\right) &\text{since $\sum_{g_j \in \Gamma(S)}w_{ij}p_j \ge Rn^4 m W$}\\
&\geq (1 + \frac{1}{Rkn^2}) \cdot (1+\frac{1}{n^3})^{-1} \cdot   \left(\sum_{g_j \in \Gamma(S)}{w_{ij}p_j} + \sum_{g_j \notin \Gamma(S)}w_{ij}p_j\right)  &\text{since $\sum_{g_j \not\in \Gamma(S)} w_{ij} p_j \le nmW$}\\
 &\geq (1 + \Omega(\frac{1}{kn^2}))\cdot m_i.
 \end{align*}
 Let $M = \prod_{i \in [n]}m_i$. Since $m_i \le (\max(2,U) W^2)^m$ we have $\log(M) \leq nm \log(UW)$ always. Also $\log M \ge 0$ initially. At any heavy $x_{\max}$ iteration, $\log(M)$ increases by a additive factor of $\log((1 + \Omega(\frac{1}{kn^2}))^{k}) \in \Omega(\frac{1}{n^2})$. Since $\log(M)$ is non-decreasing in every iteration of the algorithm (since the prices of the goods and the budgets of the agents only increase), the number of heavy $\xmax$-iterations  is $\mathcal{O}(n^3 m \log(WU))$. 
\end{proof}
We have now bounded the total number of $\xmax$-iterations by $\mathcal{O}(n^3 m \log(nm UW))$.

\subsubsection{On the Increase in the $L2$-Norm of the Surplus Vector in the $\xmax$-Iterations}

The $L2$ norm of the surplus vector is minimal for the balanced flow. We just look at the difference in the $L2$ norm of the surpluses with respect to the flows $f$ in $N_p$ and $f'$ in $N_{p'}$ (Updated flow as in $\ref{eq-1}$ - $(\ref{eq-4})$) . Note that this difference in surplus is at least as large as the difference between the $L2$ norm of the surplus vector with respect to $f$ in $N_p$ and $f''$ in $N_{p'}$ (balanced flow in $N_{p'}$ in Algorithm \ref{Program}). This suffices as we are upper bounding the difference in the $L2$ norm of the surplus vector in this section.   

In a light $\xmax$-iteration, the $L2$ norm of the surplus vector increases at most by a factor of $(1 + \mathcal{O}(\frac{1}{n^3}))$ and thus the total multiplicative increase in the $L2$ norm of the surplus vector resulting from such iterations is $(1 + \mathcal{O}(\frac{1}{n^3}))^{\mathcal{O}(n^3 m \cdot \log(nmW))} = (nmW)^{\mathcal{O}(m)}$. 

We now bound the multiplicative increase resulting from heavy $\xmax$-iterations.
Despite the more aggressive price update scheme in heavy $\xmax$-iterations, we can assure the same multiplicative increase. As in~\cite{DGM:Arrow-Debreu} we wish to prove that the ratio of the highest to the lowest surplus of the agents in $S$ is at most $1 + \mathcal{O}(\frac{k}{n}))$. One possible approach is to show that the number of distinct surpluses in $S$ is $\mathcal{O}(k)$ (in that case the ratio will be $(1 + \frac{1}{n})^{\mathcal{O}(k)} = 1 + \mathcal{O}(\frac{k}{n})$). In \cite{DGM:Arrow-Debreu}, this is relatively easy to argue, as goods and agents are in one-to-one correspondence and all agents having positive outflow to a good $g$ have same surplus (by the property of balanced flow). This immediately implies that there are at most $2k$ distinct surpluses of the agents in $S$ (additional  $k$ for agents with zero outflow that own one of the goods in $\Gamma(S)$). This argument does not hold in the general scenario as the number of goods can be much larger $\abs{S}$. However Lemma~\ref{technical} gives us a useful structure in the equality network.

\begin{lemma}
 The total multiplicative increase resulting in the $L2$ norm of the surplus vector in heavy $\xmax$-iterations is $(WU)^{\mathcal{O}(m)}$.
\end{lemma}
\begin{proof}
Let $S'$ be the set of agents in $S$ that partially own some good in $\Gamma(S)$ and $k = \abs{S'}$. Let $S''$ be the set of agents in $S$ with positive outflow. Any agent in $S'$ has a budget of at least $Rn^4 m W$ (follows from the definition of heavy $x_{\mathit{max}}$ iteration) and therefore has positive outflow (since its surplus is at most $n m W$ by Observation \ref{boundedsurplusum}). Thus $S' \subseteq S''$. By Lemma \ref{technical}, the budget of any agent in $S \setminus S'$ is at most $n m W$ and hence the total outflow from agents in $S \setminus S'$ is at most $n^2 m W$. Therefore any good in $\Gamma(S)$ must have inflow from an agent in $S'$ and hence the surplus of any agent in $S''$ is equal to the surplus of some agent in $S'$.

Let $r_1 > r_2 > \ldots > r_h$ with $h \le k$ be the distinct surplus values of the agents in $S''$. Agents in $S \setminus S''$ have no outflow and no ownership of any good in $\Gamma(S)$. Therefore $r_{i+1} \ge r_i/(1 + 1/n)$ by definition of $S$ for $1 \le i < h$. 

From $r_i \le (1 + \frac{1}{n})r_{i+1}$ for all $i$, we conclude $r_1 \le (1 + \frac{1}{n})^k r_h \le (1 + \frac{2k}{n})r_h$. Therefore for any type-1 agent $b_i$, we can claim that $r_f(b_i) < (1+\frac{2k}{n})r_h$. Let $r_{f'}$ be the surplus vector after the $\xmax$-iteration. Since there are no type-3 agents in this iteration, only the surpluses of the type-1 and type-2 agents belonging to $S''$ are affected (Agents belonging to $S \setminus S''$ have no ownership of goods in $\Gamma(S)$ and no outflow also, so their surpluses are unchanged when we change $f$ to $f'$) . Let $\delta_i$ denote the increase in the surplus of a type-1 agent $b_i$ and let $\mu_j$ denote the decrease of surplus of a type-2 agent $b_j$. Note that $\sum_{b_i \in T_1}\delta_i= \sum_{b_i \in T_2}\mu_i \leq \frac{1}{Rkn^2}\sum_{b_i \in T_1}r_f(b_i)$. Then,
 \begin{align*}
 \norm{r_{f'}}_2^2 - \norm{r_f}_2^2 &= \sum_{b_i \in T_1} ((r_{f}(b_i)+\delta_i)^2 -r_{f}(b_i)^2)  - \sum_{b_i \in T_2} (r_{f}(b_i)^2 - (r_{f}(b_i)-\mu_i)^2)\\
 &= 2\sum_{b_i \in T_1}r_f(b_i) \delta_i - 2\sum_{b_i \in T_2}r_f(b_i) \mu_i + \sum_{b_i \in T_1}\delta_i^2 + \sum_{b_i \in T_2}\mu_i^2\\
 &\leq 2r_1\sum_{b_i \in T_1}\delta_i - 2r_h\sum_{b_i \in T_2}\mu_i + \sum_{b_i \in T_1}\delta_i^2 + \sum_{b_i \in T_2}\mu_i^2\\
&\leq 2(r_1 -r_h)\sum_{b_i \in T_1}\delta_i  + 2\left(\sum_{b_i \in T_1}\delta_i\right)^2 \\
&\leq 2((1+\frac{2k}{n})r_h -r_h)\sum_{b_i \in T_1}\frac{1}{Rkn^2}\cdot r_f(b_i)  + 2\frac{1}{R^2k^2n^4}\left(\sum_{b_i \in T_1}r_f(b_i)\right)^2 \\
&\leq \frac{4}{Rn^3}\sum_{b_i \in T_1}r_f(b_i)^2 + n \cdot \frac{2}{R^2k^2n^4}\sum_{b_i \in T_1}r_f(b_i)^2,
\end{align*}
Thus $\norm{r_{f'}}_2^2 \in (1+\mathcal{O}(\frac{1}{n^3}))\norm{r_f}_2^2$. Therefore the multiplicative increase in any heavy $\xmax$- iterations is $1+\mathcal{O}(\frac{1}{n^3})$. Thus the total multiplicative increase in the $L2$ norm of the surplus vector in all heavy $\xmax$ iterations is at most $(1+ \mathcal{O}(\frac{1}{n^3}))^{\mathcal{O}(n^3 m \log(WU))} \in \mathcal{O}(WU)^{\mathcal{O}(m)}$.  
\end{proof}
 Thus the total multiplicative increase in all $\xmax$-iterations is at most $(nmUW)^{\mathcal{O}(m)}$.

\subsubsection{Balancing Iterations}
In the balancing iterations $x<x_{\max}$. First we discuss the case when $x = \min(x_{23},x_{24},x_{13},x_{2})$. Since the $L1$ norm of the surplus vector is non-increasing during such an iteration (by Observation \ref{boundedsurplusum}), the total decrease in the surplus of the type-2 agents is at least the total increase in the surpluses of the type-1 and type-3 agents. So now we quantify the decrease in the $L2$ norm of the surplus during such an iteration. Let $r_{\min}$ denote the lowest surplus of an agent in $S$ and $r_{\max}$ denote the highest surplus of a type-3 or type-4b agent. 
Notice that the highest surplus of any agent and hence any agent in $S$ is at most $e \cdot r_{\mathit{min}}$ and that $r_{\max} \le r_{\min}/(1 + 1/n)$. 

Each type-1 agent's surplus increases at most by a multiplicative factor of $1 + {1}/{Rkn^2}$. Every type-1 agent partially owns at least one good in $\Gamma(S)$ and therefore, $k$ is at least the number of type-1 agents in $B(S)$. Thus the total increase in the sum of squares of the surpluses as a result of increase in the surpluses of the type-1 agents is $\sum_{b_i \in T_1} ({1}/{Rkn^2})\cdot r_{f}(b_i)^2$ which is at most ${e^2r_{\mathit{min}}^2}/{Rn^2}$.
  
The surplus of the type-2 and the type-3 agents move closer to each other and the decrease in the former is at least as large as the increase in the latter. Therefore, the sum of their surpluses does not increase. We now quantify the decrease in the sum of squares of their surpluses. Let $\delta_i$ denote the decrease in the surplus of a type-2 agent $b_i$ and $\mu_j$ the increase in surplus of a type-3 agent $b_j$. Let $r_{f'}$ be the new surplus vector (w.r.t flow $f'$). Then the change in the sum of squares of type-2, type-3 and type-4 agents is 
   \begin{align*}
\sum_{b_i \in T_2}&((r_{f}(b_i)-\delta_i)^2 - r_{f}(b_i)^2)) + \sum_{b_i \in T_3}((r_{f}(b_i)+\mu_i)^2 - r_{f}(b_i)^2))\\ 
   &=\sum_{b_i \in T_2}(-2r_{f}(b_i)\delta_i + \delta_i^2) + \sum_{b_i \in T_3}(2r_{f}(b_i)\mu_i + \mu_i^2)) \\
   &=\sum_{b_i \in T_2}-r_{f}(b_i)\delta_i +\sum_{b_i \in T_3}r_{f}(b_i)\mu_i - \sum_{b_i \in T_2}\delta_i(r_{f}(b_i)-\delta_i) + \sum_{b_i \in T_3}\mu_i(r_{f}(b_i)+\mu_i)\\
\intertext{For any balancing iteration we have that $\min_{b_i \in T_2} (r_{f}(b_i) - \delta_i) \geq \max_{b_i \in T_3}(r_{f}(b_i) + \mu_i)$ and $\sum_{b_i \in T_2} \delta_i \geq \sum_{b_i \in T_3} \mu_i$. This implies that $\sum_{b_i \in T_2}\delta_i(r_{f}(b_i)-\delta_i) \geq \sum_{b_i \in T_3}\mu_i(r_{f}(b_i)+\mu_i)$. Notice that $r_{\min}$ is $\min_{b_i \in T_1 \cup T_2}r_{f}(b_i)$ and $r_{\max}$ is $\max_{b_i \in T_3 \cup T_{4b}}r_{f}(b_i)$. Therefore, we may continue}
 & \le -r_{\min}\sum_{b_i \in T_2}\delta_i + r_{\max}\sum_{b_i \in T_3}\mu_i\\
   &\leq -(r_{\min} - r_{\max}) \cdot \sum_{b_i \in T_2}\delta_i \leq -(r_{\min} - r_{\max}) \cdot \frac{(\sum_{b_i \in T_2}\delta_i + \sum_{b_i \in T_3}\mu_i )}{2}.
\intertext{Now, whenever $x = \mathit{min}(x_{23},x_{24},x_{13},x_2)$, $\sum_{b_i \in T_2}\delta_i + \sum_{b_i \in T_3} \mu_i \geq r_{\min} - r_{\max}$. Thus }
&\le - \frac{(r_{\min} - r_{\max})^2}{2} \leq -\frac{r_{\min}^2}{2(n+1)^2} \leq -\frac{r_{\min}^2}{4n^2}.
\end{align*}
Therefore $\norm{r_{f'}}_2^2 - \norm{r_f}_2^2 \leq \frac{e^2r_{\min}^2}{Rn^2} - \frac{r_{\min}^2}{4n^2} = - \frac{r_{\min}^2}{4n^2}$ (Recall that $R = 8 e^2$). Since $\norm{r_f}_2^2 \le n e^2r_{\min}^2$, we have
\[     \norm{r_{f'}}_2^2 \leq (1-\Omega(\frac{1}{n^3}))\norm{r_f}_2^2. \]

Now we look into the case when $x=\xeq$. We update the flow exactly the same way as in~\cite{DGM:Arrow-Debreu}. The new equality edge will involve a type-1 agent or a type-2 agent. But after we update the flow, we are in a similar situation as above, with a possibility that the surplus of a few type-1 agents may even decrease and be equal to that of a few type-3 agents. Like earlier the sum of surpluses of the agents is non-increasing here too and all arguments for the decrease in the sum of squares of the agents remain the same. The flow adjustment is exactly identical to the one in~\cite{DGM:Arrow-Debreu}.

\begin{lemma}
 The total number of balancing iterations is $\mathcal{O}(n^3\max(n,m) \log(nmUW))$
\end{lemma}

\begin{proof}
Every balancing iteration results in a multiplicative decrease of $1- \Omega(\frac{1}{n^3})$ in the $L2$ norm of the surplus. The total multiplicative increase as a result of $\xmax$-iterations is $(nmUW)^{\mathcal{O}(\max(n.m))}$. Initially $\norm{r_f}_2$ is at most $\sqrt{n (mW)^2}$ and the algorithm terminates with $\norm{r_f}_2$ being $\varepsilon$. Therefore the total number of balancing iterations is at most
\[\log_{1 - \Omega(\frac{1}{n^3})}(\frac{1}{\varepsilon} \cdot \sqrt{n (mW)^2} \cdot (nmUW)^{\mathcal{O}(m)}) =\mathcal{O}(n^3 \max(n,m) \log(n m U W)).\]
\vspace{-1.2cm}\par\end{proof}

So now we have bounded all the iterations of our algorithm.

\begin{theorem}
The total number of iterations is $\mathcal{O}(n^3 \max(n,m) \cdot \log(nWU))$.
\end{theorem}

\subsubsection{Extraction of Equilibrium Prices and Perturbation of Utilities}\label{perturbation}

In~\cite{\citeDM} is was shown for the special case ($n = m$ and $w$ the identity matrix) that once the total surplus is sufficiently small, the equality network for the current price vector $p$ is the equality network for the equilibrium price vector $\hat{p}$. The equilibrium price vector can then be extracted by solving a linear system. Darwish~\cite{DarwishMasterThesis} showed that the same approach also works for the general case.

\begin{theorem}
\label{extraction}
Consider any instance of the general linear Arrow-Debreu market with $n$ agents, $m$ goods, utility matrix $u$ and weight matrix $w$. If $p$ is a price vector such that  $\norm{r_f}$ at most ${1}/{(8\cdot (n+m)^{4(n+m)}(UW)^{3(n+m)})}$ for any balanced flow $f$ in $N_p$, then the equilibrium price vector $p^{*}$ can be determined in $\mathcal{O}((n+m)^4 \cdot \log(UW))$.
\end{theorem}


\cite{DGM:Arrow-Debreu} achieves $\mathcal{O}(n^2)$ time for determining the balanced flow in every iteration by keeping the Equality Network acyclic at every point in time in the algorithm. This improvement, after minor adaptions, also applies to the general case. \cite{DGM:Arrow-Debreu} achieves $\mathcal{O}(n^2)$ time for determining the balanced flow in every iteration by keeping the Equality Network acyclic at every point in time in the algorithm. The improvement is achieved by slightly perturbing every non-zero utility $u_{ij}$ to $\tilde{u}_{ij}(\delta)$, where $\delta > 0 $ and $\lim_{\delta \rightarrow 0^{+}}\tilde{u}_{ij}(\delta) = u_{ij}$. After performing the perturbation, they determine the solution $\tilde{p}$ (running their algorithm). Henceforth they efficiently determine the equilibrium price vector for the unperturbed utilities from $\tilde{p}$ and $N_{\tilde{p}}$. This improvement, after minor adaptions, also applies to the general case. We start with a review of~\cite{\citeDM} and refer to~\cite{\citeDM,DarwishMasterThesis} for details. 

Let $p$ be a set of prices. We extend $N_p$ to the \emph{extended equality network} $\hat{N}_p$ by adding the edges $\left\{(b_i,g_j)| w_{ij} > 0,\  i\in [n],\ j\in [m] \right\}$. We call $p$ \emph{canonical} if the extended equality network  is connected and the minimum price of a good is equal to one. As in~\cite{\citeDM,DarwishMasterThesis}, we can now make the following claim. 

\begin{lemma}
 Let $p'$ be any equilibrium price vector. There is a canonical price vector $p$, such that $E_{p'} \subseteq E_{p}$ where $E_{p'}$ and $E_p$ are the set of edges of $N_{p'}$ and $N_p$ respectively.
\end{lemma}
\begin{proof} 
 Let $f$ be the money flow corresponding to $p'$. Scaling the prices will make the minimum price equal to one. Let $K_0$ be a component containing a good with price one. As long as there is a component from $K_0$, we choose one such component, say $K$, and increase the budgets of the agents, the prices of the goods, and all flows in $K$ by a mutiplicative factor $x$ until a new equality edge emerges connecting an agent from $K$ to a good not in $K$. This reduces the number of components in $\hat{N}_{p'}$. Continuing the above operation we arrive at a canonical equilibrium price vector and the corresponding money flow.
\end{proof}
Let $p$ be the canonical equilibrium price vector obtained in this way and let $T$ be any spanning forest of the equality network $N_p$. Then for any agent $b_i$ with neighbors $g_{j_1}$ to $g_{j_k}$ in $T_p$, $p$ must satisfy
\begin{equation}
 \label{fr1}
    u_{ij_{1}} \cdot p_{j_{\ell}} - u_{ij_{\ell}} \cdot p_{j_1} =0 \quad \text{for } 2 \leq \ell \leq k.
\end{equation}

\noindent For all components $K$ of $N_p$ (not the extended equality network) the total budget has to be equal to the total value of the goods, i.e.,

\begin{equation}
 \label{fr2}
    \sum_{b_i \in K} \sum_{j \in [m]} w_{ij}p_j - \sum_{g_j \in K} \sum_{i \in [n]} w_{ij} p_j  = 0
\end{equation}

\noindent Also there is one good $g_i$ with unit price and therefore 

\begin{equation}
  \label{fr3}
    p_i = 1
\end{equation}

\noindent It turns out that the above set of equations are also sufficient for $p$ to be a canonical equilibrium price vector. 

\begin{lemma}
  \label{omar-fr}
  The system of equations in (\ref{fr1}), (\ref{fr3}), and all but one equation in (\ref{fr2}) has full rank. 
\end{lemma}

\begin{proof}
  Shown in~\cite{\citeDM,DarwishMasterThesis}. 
\end{proof}



Let $\tilde{p}$ be the canonical equilibrium price vector corresponding to the perturbed utilities. Let $\tilde{A}$ be the coefficient matrix of the system of equations in (\ref{fr1}), (\ref{fr2}) and (\ref{fr3}) with utilities in (\ref{fr1}) being the perturbed utilities $\tilde{u}$. Note that $\tilde{A}\tilde{p} = X$ where the \say{unit-vector} $X$ corresponds to the right hand side of the equations. Let $A$ be the matrix obtained by replacing every occurrence of $\tilde{u}_{ij}$ in $\tilde{A}$ by $u_{ij}$ for all $i \in [n], j \in [m]$. Notice that $A \tilde{p} = \tilde{X}$ where $\tilde{X}$ is obtained from $X$ by replacing the right hand side of the equations in (\ref{fr1})
by 
\[    (u_{ij_1} - \tilde{u}_{ij_{\ell}})p_{j_{\ell}} + (\tilde{u}_{ij_{\ell}} - u_{ij_1})p_{j_1}.\]
Also notice that $\norm{X - \tilde{X}}_{\infty} \leq \gamma'$ where $\gamma' = 2\gamma (\max(2,U)^{m-1}W^{2m-2})$
and  $\gamma = \max_{i \in [n], j \in [m]}$ $ \abs{u_{ij}-\tilde{u}_{ij}}$. 
We can compute the equilibirum price vector corresponding to the unperturbed utilities  efficiently as long as $\gamma$ is sufficiently small. In particular for $\gamma \in \nicefrac{1}{\Theta((MUW)^{4m})}$ where $M = m \cdot \max(n,m)$ we have the following lemma.

\begin{lemma}
 \label{det p from p'}
 Let $p$ be a vector such that $Ap = X$ where $X$ is the unit-vector corresponding to the right hand sides of the equations in (\ref{fr1}), (\ref{fr2}) and (\ref{fr3}). Then $p$ is also the equilibrium price vector corresponding to the unperturbed utilities. 
\end{lemma}

\begin{proof}
 Note that $A$ is an $m \times m$ integral matrix with every entry upper bounded by $\max(U,\max(n,m)W)$. Let $M = m \cdot\max(n,m)$ and $D$ be the determinant of $A$. We have  $D \leq (MUW)^m$. We have $Ap=X$ and $A\tilde{p} = \tilde{X}$. Thus we can express $p = \frac{1}{D}q$ and $\tilde{p} = \frac{1}{D}\tilde{q}$ where both $q$ and $\tilde{q}$ are integral vectors. Also we have $A(p-\tilde{p}) = X -\tilde{X}$. Therefore we have $p-\tilde{p} = A^{-1} (X - \tilde{X}) = \nicefrac{1}{D} \cdot B(X - \tilde{X})$, where $B$ is an integral matrix with largest entry at most $m!(U \max(n,m)W)^m$ . In particular, for any $i \in [m]$ we have $\abs{p_i - \tilde{p}_i } \leq \nicefrac{(MUW)^m \gamma'}{D}$. This also implies that $\lvert q_i - \tilde{q}_i \rvert \leq (MUW)^m \gamma'$. Note that to prove that $p$ is an equilibrium price vector corresponding to the unperturbed utilities, it suffices to show that the min-cut in $N_p$ is $\sum_{i \in [n]} \sum_{j \in [m]} w_{ij}p_j$. To this end we first show that the set of edges in $N_{\tilde{p}}$ is a subset of that in $N_p$.
 
 \begin{observation}
  Let $E_{\tilde{p}}$ and $E_p$ be the set of edges in $N_{\tilde{p}}$ and $N_p$ respectively. We have $E_{\tilde{p}} \subseteq E_p$.
 \end{observation} 
 
 \begin{proof}
First observe that the edges $(s,b_i)$ for all $i \in [n]$ and $(g_j,t)$ for all $j \in [m]$ belong to both $E_{\tilde{p}}$ and $E_{p}$. We now argue about the bang-for-buck edges. To this end, consider  any edge $(b_i,g_j) \in E_{\tilde{p}}$. We have $\frac{\tilde{u}_{ij}}{\tilde{p}_j} \geq \frac{\tilde{u}_{ik}}{\tilde{p}_k}$ for all $k \in [m]$. Or equivalently we have $\tilde{u}_{ij}\tilde{q}_k \geq \tilde{u}_{ik}\tilde{q}_j$ for all $k \in [m]$. Note that to show $(b_i,g_j) \in N_{p}$, it suffices to show that $u_{ij}q_k \geq u_{ik}q_j$ for all $k \in [m]$. Using that $\gamma \ll  \nicefrac{1}{\Theta((MUW)^{3m})}$ and $\gamma' \ll \nicefrac{1}{\Theta((MUW)^{2m})}$ we have 
  \begin{align*}
   u_{ij}q_k  &\geq (\tilde{u}_{ij} - \gamma)\tilde{q}_k  + (\tilde{u}_{ij} - \gamma)(q_k - \tilde{q}_k)\\
              &\geq \tilde{u}_{ij}\tilde{q}_k - \gamma \tilde{q}_k -  U(MUW)^m\gamma'\\
              &\geq \tilde{u}_{ik}\tilde{q}_j - \gamma \tilde{q}_k -  U(MUW)^m\gamma'\\
              &\geq (u_{ik}-\gamma)q_j + (u_{ik} - \gamma)(\tilde{q}_j - q_j) - \gamma \tilde{q}_k-  U(MUW)^m\gamma'\\
              &\geq u_{ik}q_j - \gamma(q_j + \tilde{q}_k) - 2U(MUW)^m\gamma'\\
              &> u_{ik}q_j - 1
 \end{align*}  
  From the integrality of the matrix $u$ and the vector $q$ we have that $u_{ij}q_k \geq u_{ik}q_j$ or equivalently $\frac{u_{ij}}{p_j} \geq \frac{u_{ik}}{p_k}$ for all $k \in [m]$. Therefore we have $E_{\tilde{p}} \subseteq E_{p}$.   
 \end{proof}
In particular, any cut $Z$ in $N_{p}$ is also a cut in $N_{\tilde{p}}$. Consider any edge $e \in E_{\tilde{p}} \cap E_{p}$. Let $\tilde{c}(e)$ denote the capacity of $e$ in $N_{\tilde{p}}$ and $c(e)$ denote the capacity of $e$ in $N_{p}$. For any cut $Z$ of $N_{p}$ and $N_{\tilde{p}}$, we define $c(Z)$ and $\tilde{c}(Z)$ analogously. We now make the following observation.

\begin{observation}
 \label{cut-bounds}
 Let $Z$ be the min-cut in $N_{p}$. We have $\abs{c(Z) - \tilde{c}(Z)} \le  \nicefrac{(n+m)^2W(MUW)^m\gamma'}{D}$.  
\end{observation}

\begin{proof}
 Consider any edge $e \in E_{\tilde{p}} \cap E_{p}$. Observe that 
 
 \begin{itemize}
  \item If $e = (b_i,g_j)$ for any $i \in [n], j \in [m]$ then $c(e) = \tilde{c}(e) = \infty$.
  \item If $e = (g_j,t)$ for any $j \in [m]$, then 
\[ \tilde{c}(e)  = \sum_{i \in [n]} w_{ij}\tilde{p}_j
                     = \sum_{i \in [n]}w_{ij}p_j + \sum_{i \in [n]}w_{ij}(\tilde{p}_j-p_j)
                     = c(e) + \sum_{i \in [n]}w_{ij}(\tilde{p}_j-p_j).\]
   Since $\abs{ p'_j - p_j} \leq \nicefrac{(MUW)^m\gamma'}{D}$ we have that $\abs{c(e) - \tilde{c}(e)} \leq  \nicefrac{nW(MUW)^m\gamma'}{D}$.       
  \item If $e = (s,b_i)$ for any $i \in [n]$, by a symmetric argument we have $\abs{c(e) - \tilde{c}(e)}\leq  \nicefrac{mW(MUW)^m\gamma'}{D}$.  
 \end{itemize}
Therefore, $\abs{c(Z) - \tilde{c}(Z)} \leq \abs{Z} \cdot \nicefrac{(n+m)W(MUW)^m\gamma'}{D}$. Since $Z$ is the min-cut it does not contain the bang per buck edges (as their capacity is $\infty$). Hence $\abs{Z}$ is at most $n+m$ and the claim follows. 
\end{proof} 
Consider the cut $Z = (s, B \cup G \cup t)$. $Z$ is a min-cut in $N_{\tilde{p}}$. Note that $c(Z) = \nicefrac{\sum_{i \in [n]} \sum_{j \in [m]} w_{ij}q_j}{D}$. Assume that $Z$ is not the min-cut in $N_{p}$ and let $Z'$ be the min-cut. Then we have $c(Z') \leq \nicefrac{(\sum_{i \in [n]} \sum_{j \in [m]} w_{ij}q'_j-1)}{D}$. We now upper bound $\tilde{c}(Z')$. We have (the first inequality uses Observation~\ref{cut-bounds})

\begin{align*}
 \tilde{c}(Z')   &\leq c(Z') + \frac{(n+m)^2W(MUW)^m\gamma'}{D} \\
                 &\leq \frac{\sum_{i \in [n]} \sum_{j \in [m]} w_{ij}q_j}{D} - \frac{1}{D} + \frac{(n+m)^2W(MUW)^m\gamma'}{D}\\ 
                 &=\frac{\sum_{i \in [n]} \sum_{j \in [m]} w_{ij}\tilde{q}_j}{D} + \frac{\sum_{i \in [n]} \sum_{j \in [m]} w_{ij}(q_j-\tilde{q}_j)}{D} - \frac{1}{D} + \frac{(n+m)^2W(MUW)^m\gamma'}{D}\\
                 &\leq \sum_{i \in [n]} \sum_{j \in [m]} w_{ij} + \sum_{i \in [n]} \sum_{j \in [m]} w_{ij}\frac{(MUW)^m\gamma'}{D}  - \frac{1}{D} + \frac{(n+m)^2W(MUW)^m\gamma'}{D}\\
                 &\leq \tilde{c}(Z) - \frac{1}{D} + 2\frac{(n+m)^2W(MUW)^m\gamma'}{D}\\
                 &< \tilde{c}(Z), 
\end{align*}
a contradiction. Thus $Z$ is a min-cut in $N_p$ as well and has capacity $\sum_{i \in [n]} \sum_{ j \in [m]} w_{ij}p_{j}$.
\end{proof}
By the above, the equilibrium price vector corresponding to the unperturbed utilities can be determined from $\tilde{p}$ and $N_{\tilde{p}}$ by solving the system $Ap =b$ where $A$ has full rank and the absolute value of every entry is bounded by $\max(U,\max(n,m)W)$. This can be achieved in $\mathcal{O}(m^3)$ arithmetic operations and in time $\mathcal{O}(m^4 \cdot \log(nmUW))$. Therefore we can use the same perturbation scheme used in~\cite{DGM:Arrow-Debreu} and ensure that the equality network is acyclic at every iteration of our algorithm. This ensures that we can determine the balanced flow in $\mathcal{O}((n+m)^2)$.

\subsection{Summary}

\begin{theorem}
 The market clearing price vector for the general Arrow-Debreu market can be determined with $\mathcal{O}((n+m)^6 \log (nmWU))$ arithmetic operations.
\end{theorem}

\begin{proof}
 We perturb the utility matrix following the same perturbation as in \cite{DGM:Arrow-Debreu}. This perturbation ensures that the equality network is acyclic at any point in time in the algorithm. Thereafter we run Algorithm \ref{Program} until $\norm{r_f} < {1}/{(8\cdot (n+m)^{4(n+m)}(UW)^{3(n+m)})}$. This involves $\mathcal{O}(n^3 m$ $ \log (nWU))$ iterations. In each iteration we can determine $x_{23}$, $x_{24}$, $x_{13}$, $x_2$ and $\xeq$ in $\mathcal{O}(n^2)$ comparisons. The balanced flow can also be determined by $n+m$ max flow calls in $N_p$. Since $N_p$ is acyclic (due to the perturbation), we can compute each max flow in $\mathcal{O}(n+m)$ arithmetic operations as in \cite{DGM:Arrow-Debreu}. Thus every iteration involves $\mathcal{O}((n+m)^2)$ arithmetic operations and comparisons. Therefore Algorithm 1 terminates performing $\mathcal{O}((n+m)^6 \log(nmWU))$ arithmetic operations and comparisons. Thereafter we perform extraction as in Theorem $\ref{extraction}$ in time $\mathcal{O}((n+m)^4 \log(nmUW))$ and determine the equilibrium prices for the perturbed utilities. We then determine equilirbium prices corresponding to the original utilities from the equilibrium prices of the perturbed utilities performing $\mathcal{O}(m^3)$ arithmetic operations in $\mathcal{O}(m^4 \log(UW))$ time. Overall we perform $\mathcal{O}((n+m)^6 \log(nmWU))$ arithmetic operations and comparisons.   
\end{proof}
 
To achieve the polynomial running time we can follow the same strategy used in~\cite{\citeDM}. where we restrict the prices and the update factor to powers of $1 + 1/L$ where $L$ has polynomial bit length (linear in $n+m$). This guarantees that all arithmetic is done on rationals of polynomial bitlength. This can be adapted to the perturbation as well~\cite{DGM:Arrow-Debreu}.

\section{Lower Bounds for the Algorithms in~\cite{Duan-Mehlhorn:Arrow-Debreu-Market,DGM:Arrow-Debreu}} 
In this section, we construct non-trivial instances that make the equilibrium prices exponential in $U$ and forces the algorithms in~\cite{Duan-Mehlhorn:Arrow-Debreu-Market,DGM:Arrow-Debreu} to execute large number of iterations. 

\subsection{Instances where Algorithm in~\cite{DGM:Arrow-Debreu} requires $\Omega(n^4 \log U)$ Iterations}
 \label{DGMI_n}


We construct an $I_n$ comprises of a set of agents $B = \left\{b_1,b_2,...,b_n\right\}$ and a set of goods $G = \left\{g_1,g_2,...,g_n\right\}$, and $n$ is $\emph{even}$. There is exactly one unit of each good ($W =1$) and agent $b_i$ only owns one unit of good $g_i$. We now define the utility matrix as follows: $u_{i,i-1} = U \quad  \text{for } 2 \leq i \leq n$, $u_{i,i+1} = 1 \quad  \text{for } 2 \leq i \leq n-1$, and $u_{1,1} = u_{1,2} = U$. 
$u_{i,j} = 0$ for every other pair of $(i,j)$. See Figure~\ref{First Example}.

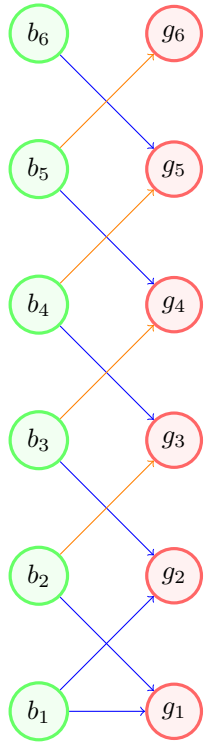
\begin{figure}[t]
\centering
\begin{tikzpicture}
[
roundnode/.style={circle, draw=green!60, fill=green!5, very thick, minimum size=7mm},
roundnode2/.style={circle, draw=red!60, fill=red!5, very thick, minimum size=7mm},
]

\node[roundnode]      (l6)       {$b_6$};
\node[roundnode]      (l5)       [below=of l6] {$b_5$};
\node[roundnode]      (l4)       [below=of l5] {$b_4$};
\node[roundnode]      (l3)       [below=of  l4]{$b_3$};
\node[roundnode]      (l2)       [below=of  l3]{$b_2$};
\node[roundnode]      (l1)       [below=of  l2]{$b_1$};

\node[roundnode2]      (r6)       [right=of l6] {$g_6$};
\node[roundnode2]      (r5)       [right=of l5] {$g_5$};
\node[roundnode2]      (r4)       [right=of l4] {$g_4$};
\node[roundnode2]      (r3)       [right=of l3] {$g_3$};
\node[roundnode2]      (r2)       [right=of l2] {$g_2$};
\node[roundnode2]      (r1)       [right=of l1] {$g_1$};
 
\draw[blue,->] (l6) -- (r5);
\draw[blue,->] (l5) -- (r4);
\draw[blue,->] (l4) -- (r3);

\draw[blue,->] (l3) -- (r2);
\draw[blue,->] (l2) -- (r1);
\draw[blue,->] (l1) -- (r2);
\draw[blue,->] (l1) -- (r1);


\draw[orange,->] (l2) -- (r3);
\draw[orange,->] (l3) -- (r4);
\draw[orange,->] (l4) -- (r5);
\draw[orange,->] (l5) -- (r6);

\end{tikzpicture}    
\caption{Utility matrix for $n =6$. The agents are the green nodes on the left and the goods are the red nodes on the right. Agent $b_i$ owns good $g_i$. The blue edges represent utility $U$ and the orange ones represent utility one}\label{First Example}\end{figure}


\begin{observation}
 \label{irreducibleI_n}
  The instance $I_n$ is irreducible or equivalently for every $P \subset B$, there exists $b_i \in P$ ,$b_j \notin P$ such that $u_{i,j}>0$.  
\end{observation}

\begin{proof}
For contradiction let us assume that there exists a $P \subset B$, such that the agents in $P$ are only interested in the goods they own. If $P$ contains $b_1$ then it must contain $b_2$. The agent $b_i$ is interested in the good $g_{i+1}$ and hence $b_i \in P$ implies  $b_{i+1} \in P$. Thus $P =B$, which is a contradiction. Hence $P$ cannot contain $b_1$. If $P$ contains any agent $b_i \neq b_1$ then $P$ must also contain $b_{i-1}$, implying that $P$ must contain $b_1$, which is again a contradiction.
\end{proof}

\begin{theorem}
 For the instance $I_n$, we have the ratio of the maximum to the minimum price of a good at equilirbium is $\Omega(U^{\Omega(n)})$. 
\end{theorem}

\begin{proof}
 Let $p$ be the market clearing price vector with $p_i$ denoting the price of good $g_i$, and $f$ be the money flow at equilibrium. Since the only agent interested in $g_n$ is $b_{n-1}$ , $p_{n-1} \geq p_n$. We now discuss  two disjoint scenarios,
 \begin{itemize}
    \item $p_n = p_{n-1}$. In this case we claim that for every even $i$, $p_{i-1} = p_i \geq U \cdot p_{i+1} = U \cdot p_{i+2}$ and $f_{i,i-1} = f_{i-1,i} = p_i = p_{i-1}$. For the base case, $i=n  $ we have $p_n = p_{n-1}$ and $f_{n,n-1}=f_{n-1,n}=p_n=p_{n-1}$. For the inductive step we assume that the claim holds for $i+2$. Since $g_{i+2}$ is a bang per buck good for $b_{i+1}$, $p_i \geq U\cdot p_{i+2} = U \cdot p_{i+1}$. Since the only other agent interested in $g_i$ is $b_{i-1}$ we may conclude that $p_{i-1} \geq p_i$ ($b_{i-1}$ is the only agent investing in $g_i$). But then again, since the only good $b_i$ invests in is $g_{i-1}$, $p_{i-1} \leq p_i$. This implies that $p_{i-1}=p_i\geq U \cdot p_{i+1} = U \cdot p_{i+2}$ and $f_{i,i-1} = f_{i-1,i} = p_i = p_{i-1}$.
    \item $p_{n-1} > p_n$. In this case we claim that for every even $i$, $p_i < p_{i-1}$ and $p_{i-1} \geq U \cdot p_{i+1}$. For the base case $i=n$, this trivially holds. For the inductive step we assume that our claim holds for $i+2$. Since $p_{i+1} > p_{i+2}$, the agent $b_{i+1}$ must invest in the good $g_{i}$. This implies that $p_{i} \leq U \cdot p_{i+2}$. Since $p_{i+2} < p_{i+1}$, agent $b_i$ must invest in good $g_{i+1}$. Therefore $g_{i+1}$ is a bang per buck good for agent $b_i$, implying that $p_{i-1} \geq U \cdot p_{i+1} > U \cdot p_{i+2} \geq p_i$.
   \end{itemize} 
 In either case, we have $p_{i-1} \ge U p_{i+1}$ for even $i$. Thus there are goods with price ratio equal to $U^{\frac{n}{2}-1}$.
\end{proof}
Note that the algorithm in~\cite{Duan-Mehlhorn:Arrow-Debreu-Market} never increases the price of a good more than a multiplicative factor of $(1 + \mathcal{O}(1/n^3))$. Thus the number of iterations is $\log_{1 + \mathcal{O}(1/Rn^3)} U^{\Omega(n)} \in \Omega(n^4 \log(U))$. To prove that the algorithm in~\cite{Duan-Mehlhorn:Arrow-Debreu-Market,DGM:Arrow-Debreu} takes $\Omega(n^4 \log(U))$ iterations, we first need to understand the details of how the sets $S$ (high surplus agents) and $\Gamma(S)$ (high demand goods) evolve throughout the iterations of the algorithm in~\cite{Duan-Mehlhorn:Arrow-Debreu-Market,DGM:Arrow-Debreu}. In particular we show that there exists a good $g_i$ and its price is increased by a multiplicative factor of $U^{\Omega(n)}$ in $\xmax$ iterations with $\Omega(n)$ type-1 agents. Note that the price of any good is increased at most by a multiplicative factor of $1 + \mathcal{O}(1/n^3)$ in any $\xmax$ iteration with $\Omega(n)$ type-1 agents. Since the total price increase in such iterations is $U^{\Omega(n)}$, the number of such iterations is $\Omega(\log_{1 + \mathcal{O}(1/Rn^3)} U^{\Omega(n)}) \in \Omega(n^4 \log(U))$. 

We first give details on how the algorithm in~\cite{Duan-Mehlhorn:Arrow-Debreu-Market} operates on $I_n$. This will also elucidate why the algorithm in~\cite{DGM:Arrow-Debreu} and our algorithm will also incur the same running time. We start with some observations:

\begin{observation}
 \label{technical-DGM}
 Let $n$ be an even integer and  $p_i$ be the price of $g_i$ and the budget of $b_i$ for $i \in [n]$. For an even integer $k < n $,
 \begin{itemize}
   \item $p_i = 1 \quad  k < i \leq n $,
   \item $1 \leq p_k = p_{k-1} \leq p_{k-2} = p_{k-3} \leq ...... p_4=p_3 \leq p_2=p_1$.
 \end{itemize}
 The equality network $N_p$ comprises of the following edges in addition to the ones connecting the agents to the source and the goods to the sink,
 \begin{itemize}
  \item $(b_i,g_{i-1}) \quad 2 \leq i \leq n$,
  \item $(b_i,g_{i+1}) \quad 1 \leq i \leq k-1$,
  \item $(b_1,g_1)$.
 \end{itemize}
 Then, the surplus $r_f(b_i)$ for an agent $b_i$ w.r.t a balanced flow $f$ is
 \begin{itemize}
  \item $r_f(b_i) = 0  \quad  k+1<i \leq n$
  \item $r_f(b_i) = \frac{1}{k+1} \quad 1 \leq i \leq k+1$.
 \end{itemize}
\end{observation}

\begin{proof}

Note that $(s \cup B \cup \Gamma(B),\left\{ g_n,t \right\})$ is a cut in the equality network with capacity $\sum_{i\in [n-1]}^{}p_i$. Any flow that saturates the edges of this cut is a maximum flow. We give the definition of such a flow $f$:

\begin{itemize}
    \item $f_{s,i} = p_i \quad i \in [n] \setminus \left\{k+1\right\}$,
    \item $f_{i,t} = p_i \quad i \in [n-1]$,
    \item $f_{i,i-1} = 1 \quad k+1 < i \leq n$,
    \item $f_{k+1,k} = 0$,
    \item $f_{i,i-1} = f_{i-1,i} = p_i = p_{i-1} \quad \text{for all even } i \in [k]$,
    \item $f_{1,1} = 0$.
\end{itemize}
Now we show how to transform this flow to get a balanced flow through a series of augmentations. We denote the residual equality network with respect to the flow $f$ as $N_p(f)$. We refer to all the directed edges of the equality network as \emph{forward edges} in $N_p(f)$, and the reverse of the edges carrying positive flow as \emph{backward edges} in $N_p(f)$. Capacity of any {forward edge} is $\infty$ and that of the {backward edge }$e$ is $f(\mathit{rev}(e))$. The maximum flow $f$ is balanced iff there exists no path from any $b_i$ to $b_j$, in $N_p(f)$ where $r_f(b_i) > r_f(b_j)$.

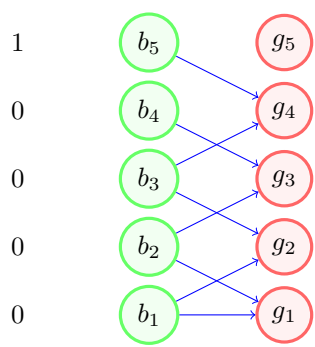
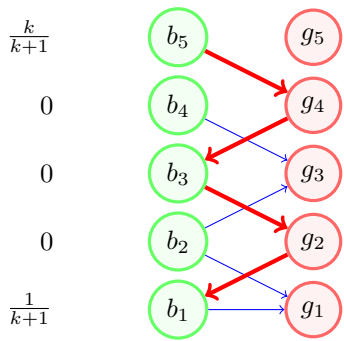
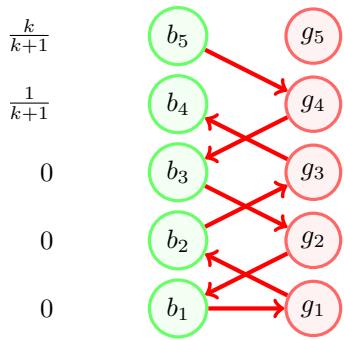
\begin{figure}[thb!]
 
 \begin{subfigure}{\textwidth}

\begin{tikzpicture}
[
roundnode/.style={circle, draw=green!60, fill=green!5, very thick, minimum size=7mm},
roundnode2/.style={circle, draw=red!60, fill=red!5, very thick, minimum size=7mm},
whitenodes/.style={circle, draw=white!60, fill=white!55, very thick, minimum size=5mm},
roundnode4/.style={circle, draw=red!60, fill=red!55, very thick, minimum size=5mm},
]

\node[whitenodes]     (w) {};

\node[roundnode]      (l5)       [right=160pt of w] {$b_5$};
\node[roundnode]      (l4)       [below=3pt of l5] {$b_4$};
\node[roundnode]      (l3)       [below=3pt of  l4]{$b_3$};
\node[roundnode]      (l2)       [below=3pt of  l3]{$b_2$};
\node[roundnode]      (l1)       [below=3pt of  l2]{$b_1$};

\node[roundnode2]      (r5)       [right=of l5] {$g_5$};
\node[roundnode2]      (r4)       [right=of l4] {$g_4$};
\node[roundnode2]      (r3)       [right=of l3] {$g_3$};
\node[roundnode2]      (r2)       [right=of l2] {$g_2$};
\node[roundnode2]      (r1)       [right=of l1] {$g_1$};

\node[whitenodes]      (s1)       [left=of l1] {$0$};
\node[whitenodes]      (s2)       [left=of l2] {$0$};
\node[whitenodes]      (s3)       [left=of l3] {$0$};
\node[whitenodes]      (s4)       [left=of l4] {$0$};
\node[whitenodes]      (s5)       [left=of l5] {$1$};


\draw[blue,->] (l5) -- (r4);
\draw[blue,->] (l4) -- (r3);
\draw[blue,->] (l3) -- (r2);
\draw[blue,->] (l2) -- (r1);
\draw[blue,->] (l1) -- (r1);
\draw[blue,->] (l1) -- (r2);
\draw[blue,->] (l2) -- (r3);
\draw[blue,->] (l3) -- (r4);

\end{tikzpicture}
   \caption{The equality network $N_p$ and the surplus vector with respect to $f$. Here $k=4$.} \label{fig:4a}
 \end{subfigure}

 \begin{subfigure}{\textwidth}

\begin{tikzpicture}
[
roundnode/.style={circle, draw=green!60, fill=green!5, very thick, minimum size=7mm},
roundnode2/.style={circle, draw=red!60, fill=red!5, very thick, minimum size=7mm},
whitenodes/.style={circle, draw=white!60, fill=white!55, very thick, minimum size=5mm},
roundnode4/.style={circle, draw=red!60, fill=red!55, very thick, minimum size=5mm},
]

\node[whitenodes]     (w) {};

\node[roundnode]      (l5)       [right=160pt of w] {$b_5$};
\node[roundnode]      (l4)       [below=3pt of l5] {$b_4$};
\node[roundnode]      (l3)       [below=3pt of  l4]{$b_3$};
\node[roundnode]      (l2)       [below=3pt of  l3]{$b_2$};
\node[roundnode]      (l1)       [below=3pt of  l2]{$b_1$};

\node[roundnode2]      (r5)       [right=of l5] {$g_5$};
\node[roundnode2]      (r4)       [right=of l4] {$g_4$};
\node[roundnode2]      (r3)       [right=of l3] {$g_3$};
\node[roundnode2]      (r2)       [right=of l2] {$g_2$};
\node[roundnode2]      (r1)       [right=of l1] {$g_1$};

\node[whitenodes]      (s1)       [left=of l1] {$\frac{1}{k+1}$};
\node[whitenodes]      (s2)       [left=of l2] {$0$};
\node[whitenodes]      (s3)       [left=of l3] {$0$};
\node[whitenodes]      (s4)       [left=of l4] {$0$};
\node[whitenodes]      (s5)       [left=of l5] {$\frac{k}{k+1}$};


\draw[blue,->] (l4) -- (r3);
\draw[blue,->] (l2) -- (r1);
\draw[blue,->] (l1) -- (r1);
\draw[blue,->] (l2) -- (r3);

\draw[red,ultra thick,->] (l5) -- (r4);
\draw[red,ultra thick,->] (r4) -- (l3);
\draw[red,ultra thick,->] (l3) -- (r2);
\draw[red,ultra thick,->] (r2) -- (l1);

\end{tikzpicture}
   \caption{Augmenting flow of $\frac{1}{k+1}$ along a path from $b_{5}$ to $b_1$ (odd $i$).} \label{fig:4b}
 \end{subfigure}
 
 \begin{subfigure}{\textwidth}

\begin{tikzpicture}
[
roundnode/.style={circle, draw=green!60, fill=green!5, very thick, minimum size=7mm},
roundnode2/.style={circle, draw=red!60, fill=red!5, very thick, minimum size=7mm},
whitenodes/.style={circle, draw=white!60, fill=white!55, very thick, minimum size=5mm},
roundnode4/.style={circle, draw=red!60, fill=red!55, very thick, minimum size=5mm},
]

\node[whitenodes]     (w) {};
\node[roundnode]      (l5)       [right=160pt of w] {$b_5$};
\node[roundnode]      (l4)       [below=3pt of l5] {$b_4$};
\node[roundnode]      (l3)       [below=3pt of  l4]{$b_3$};
\node[roundnode]      (l2)       [below=3pt of  l3]{$b_2$};
\node[roundnode]      (l1)       [below=3pt of  l2]{$b_1$};

\node[roundnode2]      (r5)       [right=of l5] {$g_5$};
\node[roundnode2]      (r4)       [right=of l4] {$g_4$};
\node[roundnode2]      (r3)       [right=of l3] {$g_3$};
\node[roundnode2]      (r2)       [right=of l2] {$g_2$};
\node[roundnode2]      (r1)       [right=of l1] {$g_1$};

\node[whitenodes]      (s1)       [left=of l1] {$0$};
\node[whitenodes]      (s2)       [left=of l2] {$0$};
\node[whitenodes]      (s3)       [left=of l3] {$0$};
\node[whitenodes]      (s4)       [left=of l4] {$\frac{1}{k+1}$};
\node[whitenodes]      (s5)       [left=of l5] {$\frac{k}{k+1}$};



\draw[red,ultra thick,->] (l5) -- (r4);
\draw[red,ultra thick,->] (r4) -- (l3);
\draw[red,ultra thick,->] (l3) -- (r2);
\draw[red,ultra thick,->] (r2) -- (l1);
\draw[red,ultra thick,->] (l1) -- (r1);
\draw[red,ultra thick,->] (r1) -- (l2);
\draw[red,ultra thick,->] (l2) -- (r3);
\draw[red,ultra thick,->] (r3) -- (l4);

\end{tikzpicture}
   \caption{Augmenting flow of $\frac{1}{k+1}$ along a path from $b_{5}$ to $b_4$ (even $i$).} \label{fig:4a}
 \end{subfigure}
 \caption{The two different augmentations transforming the maximum flow $f$ into a balanced flow.}
\label{Transformation to balanced flow}
\end{figure}

In our definition of $f$, we see that capacity of any {backward edge} in $N_p(f)$ is at least 1. We also observe that $r_f(b_{k+1}) = 1$ and $r_f(b_i) = 0$ for all $i \neq k$. The following procedure converts the maximum flow $f$ into a balanced flow. 

\begin{enumerate}
  \item For every odd $i \in [k]$ augment a flow of $\frac{1}{k+1}$ along the path $b_{k+1} \rightarrow g_k \rightarrow b_{k-1} \rightarrow g_{k-2} \rightarrow ....\rightarrow g_{i+1} \rightarrow b_i$ in $N_p(f)$. It can be verified easily that this is a valid path in the residual equality network.
  This augmentation increases the surplus of agent $b_i$ from $0$ to $\frac{1}{k+1}$ and decreases the surplus of $b_{k+1}$ by $\frac{1}{k+1}$. See Figure~\ref{Transformation to balanced flow}b.
  
  \item For every even $i \in [1,k]$ augment a flow of $\frac{1}{k+1}$ along the path $b_{k+1} \rightarrow g_k \rightarrow b_{k-1} \rightarrow g_{k-2} \rightarrow ....\rightarrow g_{2} \rightarrow b_1 \rightarrow g_1 \rightarrow b_2 \rightarrow g_3 \rightarrow .... \rightarrow g_{i-1} \rightarrow b_i$ in $N_p(f)$. Like earlier,it can be verified easily that this is a valid path in the residual equality network. This augmentation increases the surplus of agent $b_i$ from $0$ to $\frac{1}{k+1}$ and decreases the surplus of $b_{k+1}$ by $\frac{1}{k+1}$. See Figure~\ref{Transformation to balanced flow}c.
\end{enumerate}

The total flow augmented along any edge is at most $\frac{k}{k+1} < 1$. Since the backward edges of $N_p(f)$ had capacity $1$, they always remain in the residual network in all the flow augmentations. Thus the above augmentations are valid and there would be a total of $k$ such augmentations. After the augmentations, we have  $r_f(b_i) = \frac{1}{k+1}$  for all $i \in [k+1]$ and $r_f(b_i) = 0$ otherwise. This is a balanced flow.
\end{proof}
\emph{Notice that in the above lemma the surplus remains the same, as long as the price vector satisfies the necessary condition as mentioned in the observation, irrespective of the exact price of the goods.}

Now we give the reader an impression of how the price vector changes (See Figure~\ref{price change}). The algorithm in~\cite{Duan-Mehlhorn:Arrow-Debreu-Market} initializes all the prices to unity at the beginning. The equality network and the price vector satisfy all the conditions mentioned in Observation \ref{technical-DGM}, with $k=2$. $\Gamma(S) = \left\{g_1,g_2\right\}$. As a consequence of Observation \ref{technical-DGM}, we can say that the increase in prices of $\left\{g_1,g_2\right\}$ does not change the surplus vector as long as the equality network does not change. Thus the algorithm keeps increasing the prices until $p_1=p_2=U$ and the edges $(b_2,g_3)$ and $(b_3,g_4)$ appear. The new equality network and the updated price vector again satisfy the conditions in Observation \ref{technical-DGM}, with $k=4$ and $\Gamma(S) =\left\{g_1,g_2,g_3,g_4\right\}$. Like earlier the surplus vector does not change as long as the equality network does not change. The algorithm increases the prices of $\left\{g_1,g_2,g_3,g_4\right\}$ by a multiplicative factor of $U$ and this procedure continues. Henceforth we can make the following observation,

\begin{observation}
 \label{only-type-1}
  Throughout the iterations of the algorithm in~\cite{Duan-Mehlhorn:Arrow-Debreu-Market} on instance $I_n$, there are no type-3 agents. When the equality Network changes for the $i^{th}$ time, $S = \left\{b_1,b_2,...,b_{2i+1}\right\}$ and $\left\{b_1,b_2,...,b_{2i}\right\}$ are type-1 agents and $b_{2i+1}$ is a type-2 agent.   
\end{observation}

\begin{figure}[thb!]
 
 \begin{subfigure}{1.0\textwidth}

\begin{tikzpicture}
[
roundnode/.style={circle, draw=green!60, fill=green!5, very thick, minimum size=7mm},
roundnode2/.style={circle, draw=red!60, fill=red!5, very thick, minimum size=7mm},
whitenodes/.style={circle, draw=white!60, fill=white!55, very thick, minimum size=5mm},
roundnode4/.style={circle, draw=red!60, fill=red!55, very thick, minimum size=5mm},
]

\node[whitenodes]     (w) {};
\node[roundnode]      (l6)       [right=160pt of w]{$b_6$};
\node[roundnode]      (l5)       [below=3pt of l6] {$b_5$};
\node[roundnode]      (l4)       [below=3pt of l5] {$b_4$};
\node[roundnode]      (l3)       [below=3pt of  l4]{$b_3$};
\node[roundnode]      (l2)       [below=3pt of  l3]{$b_2$};
\node[roundnode]      (l1)       [below=3pt of  l2]{$b_1$};

\node[roundnode2]      (r6)       [right=of l6] {$g_6$};
\node[roundnode2]      (r5)       [right=of l5] {$g_5$};
\node[roundnode2]      (r4)       [right=of l4] {$g_4$};
\node[roundnode2]      (r3)       [right=of l3] {$g_3$};
\node[roundnode4]      (r2)       [right=of l2] {$g_2$};
\node[roundnode4]      (r1)       [right=of l1] {$g_1$};

\draw[blue,->] (l6) -- (r5);
\draw[blue,->] (l5) -- (r4);
\draw[blue,->] (l4) -- (r3);
\draw[blue,->] (l3) -- (r2);

\draw[blue,->] (l2) -- (r1);
\draw[blue,->] (l1) -- (r2);
\draw[blue,->] (l1) -- (r1);

\node[whitenodes]  (s1)  [left=of l1] {$\frac{1}{3}$};
\node[whitenodes]  (s2)  [left=of l2] {$\frac{1}{3}$};
\node[whitenodes]  (s3)  [left=of l3] {$\frac{1}{3}$};
\node[whitenodes]  (s4)  [left=of l4] {$0$};
\node[whitenodes]  (s5)  [left=of l5] {$0$};
\node[whitenodes]  (s6)  [left=of l6] {$0$};

\node[whitenodes]  (p1)  [right=of r6] {$1$};
\node[whitenodes]  (p2)  [right=of r5] {$1$};
\node[whitenodes]  (p3)  [right=of r4] {$1$};
\node[whitenodes]  (p4)  [right=of r3] {$1$};
\node[whitenodes]  (p5)  [right=of r2] {$1$};
\node[whitenodes]  (p6)  [right=of r1] {$1$};

\end{tikzpicture}    
   \caption{During the first iteration. } \label{fig:4a}
 \end{subfigure}
 \hspace*{\fill}
 
 \begin{subfigure}{1.0\textwidth}

\begin{tikzpicture}
[
roundnode/.style={circle, draw=green!60, fill=green!5, very thick, minimum size=7mm},
roundnode2/.style={circle, draw=red!60, fill=red!5, very thick, minimum size=7mm},
roundnode4/.style={circle, draw=red!60, fill=red!55, very thick, minimum size=5mm},
whitenodes/.style={circle, draw=white!60, fill=white!55, very thick, minimum size=5mm},
]

\node[whitenodes]     (w) {};
\node[roundnode]      (l6)       [right=160pt of w]{$b_6$};
\node[roundnode]      (l5)       [below=3pt of l6] {$b_5$};
\node[roundnode]      (l4)       [below=3pt of l5] {$b_4$};
\node[roundnode]      (l3)       [below=3pt of  l4]{$b_3$};
\node[roundnode]      (l2)       [below=3pt of  l3]{$b_2$};
\node[roundnode]      (l1)       [below=3pt of  l2]{$b_1$};

\node[roundnode2]      (r6)       [right=of l6] {$g_6$};
\node[roundnode2]      (r5)       [right=of l5] {$g_5$};
\node[roundnode4]      (r4)       [right=of l4] {$g_4$};
\node[roundnode4]      (r3)       [right=of l3] {$g_3$};
\node[roundnode4]      (r2)       [right=of l2] {$g_2$};
\node[roundnode4]      (r1)       [right=of l1] {$g_1$};

\draw[blue,->] (l6) -- (r5);
\draw[blue,->] (l5) -- (r4);
\draw[blue,->] (l4) -- (r3);
\draw[blue,->] (l3) -- (r2);

\draw[blue,->] (l2) -- (r1);
\draw[blue,->] (l1) -- (r2);
\draw[blue,->] (l1) -- (r1);

\draw[blue,->] (l2) -- (r3);
\draw[blue,->] (l3) -- (r4);

\node[whitenodes]  (s1)  [left=of l1] {$\frac{1}{5}$};
\node[whitenodes]  (s2)  [left=of l2] {$\frac{1}{5}$};
\node[whitenodes]  (s3)  [left=of l3] {$\frac{1}{5}$};
\node[whitenodes]  (s4)  [left=of l4] {$\frac{1}{5}$};
\node[whitenodes]  (s5)  [left=of l5] {$\frac{1}{5}$};
\node[whitenodes]  (s6)  [left=of l6] {$0$};

\node[whitenodes]  (p1)  [right=of r6] {$1$};
\node[whitenodes]  (p2)  [right=of r5] {$1$};
\node[whitenodes]  (p3)  [right=of r4] {$1$};
\node[whitenodes]  (p4)  [right=of r3] {$1$};
\node[whitenodes]  (p5)  [right=of r2] {$U$};
\node[whitenodes]  (p6)  [right=of r1] {$U$};

\end{tikzpicture}    
   \caption{The equality network changes for the first time. Edges $(b_2,g_3)$ and $(b_3,g_4)$ appear.} \label{fig:4b}
 \end{subfigure}
 \hspace*{\fill}
 
 \begin{subfigure}{1.0\textwidth}

\begin{tikzpicture}
[
roundnode/.style={circle, draw=green!60, fill=green!5, very thick, minimum size=7mm},
roundnode2/.style={circle, draw=red!60, fill=red!5, very thick, minimum size=7mm},
whitenodes/.style={circle, draw=white!60, fill=white!55, very thick, minimum size=5mm},
]

\node[whitenodes]     (w) {};
\node[roundnode]      (l6)       [right=160pt of w] {$b_6$};
\node[roundnode]      (l5)       [below=3pt of l6] {$b_5$};
\node[roundnode]      (l4)       [below=3pt of l5] {$b_4$};
\node[roundnode]      (l3)       [below=3pt of  l4]{$b_3$};
\node[roundnode]      (l2)       [below=3pt of  l3]{$b_2$};
\node[roundnode]      (l1)       [below=3pt of  l2]{$b_1$};

\node[roundnode2]      (r6)       [right=of l6] {$g_6$};
\node[roundnode2]      (r5)       [right=of l5] {$g_5$};
\node[roundnode2]      (r4)       [right=of l4] {$g_4$};
\node[roundnode2]      (r3)       [right=of l3] {$g_3$};
\node[roundnode2]      (r2)       [right=of l2] {$g_2$};
\node[roundnode2]      (r1)       [right=of l1] {$g_1$};

\draw[blue,->] (l6) -- (r5);
\draw[blue,->] (l5) -- (r4);
\draw[blue,->] (l4) -- (r3);
\draw[blue,->] (l3) -- (r2);

\draw[blue,->] (l2) -- (r1);
\draw[blue,->] (l1) -- (r2);
\draw[blue,->] (l1) -- (r1);

\draw[blue,->] (l2) -- (r3);
\draw[blue,->] (l3) -- (r4);

\draw[blue,->] (l5) -- (r6);

\draw[blue,->] (l4) -- (r5);

\node[whitenodes]  (s1)  [left=of l1] {$0$};
\node[whitenodes]  (s2)  [left=of l2] {$0$};
\node[whitenodes]  (s3)  [left=of l3] {$0$};
\node[whitenodes]  (s4)  [left=of l4] {$0$};
\node[whitenodes]  (s5)  [left=of l5] {$0$};
\node[whitenodes]  (s6)  [left=of l6] {$0$};

\node[whitenodes]  (p1)  [right=of r6] {$1$};
\node[whitenodes]  (p2)  [right=of r5] {$1$};
\node[whitenodes]  (p3)  [right=of r4] {$U$};
\node[whitenodes]  (p4)  [right=of r3] {$U$};
\node[whitenodes]  (p5)  [right=of r2] {$U^2$};
\node[whitenodes]  (p6)  [right=of r1] {$U^2$};

\end{tikzpicture}    
   \caption{The equality network changes for the third time. Edges $(b_4,g_5)$ and $(b_5,g_6)$ appear.} \label{fig:4c}
 \end{subfigure}
 \hspace*{\fill}
 \caption{Running the $\DM$ algorithm on the instance. The numbers on the left represent surplus of the agents and that on the right the prices of the goods. The red nodes on the right represent the set $\Gamma(S)$ in every iteration.}\label{price change}
\end{figure}
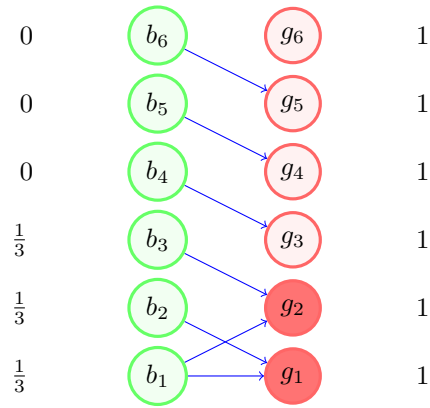
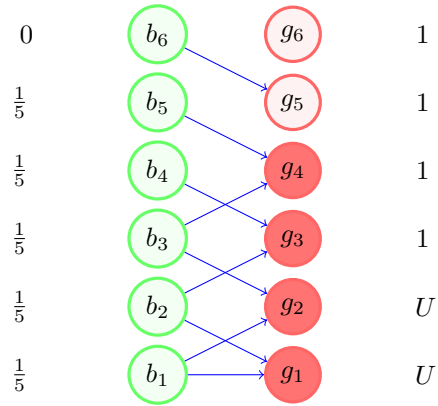
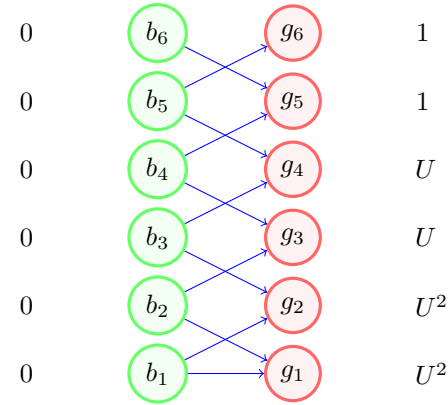

The algorithm in~\cite{DGM:Arrow-Debreu} achieves a factor $n$ improvement on the number of iterations. This is achieved by a careful selection of the set $S$ (which is identical to the one presented in this paper earlier) and a modified way of updating the prices as follows:

\[ x_{max} = \begin{cases} 
      1+ \frac{1}{Rn^3} &  \text{if there are type-3 agents}\\
      1+ \frac{1}{Rkn^2} & \text{otherwise, where $k$ = number of type-1 agents}.
   \end{cases} \]

Now we show that the analysis of the number of iterations of~\cite{DGM:Arrow-Debreu} is tight upto a factor of $\log(n)$.

\begin{theorem}
   The algorithm in~\cite{DGM:Arrow-Debreu} requires $\Omega(n^4 \log U)$ iterations on instance $I_n$. 
\end{theorem}
\begin{proof}
   Note that the equality network changes exactly like it did while running the algorithm in~\cite{Duan-Mehlhorn:Arrow-Debreu-Market}. At every stage of the iteration we can partition the agents $B$ into $B_1 = \left\{b_i|i \in [k+1]\right\}$ and $B_2 = B \backslash B_1$. With all agents in $B_1$ having the same surplus equal to ${1}/(k+1)$ and agents in $B_2$ having zero surplus. Thus there exists no agent $b_j$, such that $r_f(b_l) > r_f(b_j) \geq r_f(b_l)/(1+1/n)$, when $r_f(b_l) \in S$. This proves that the set $S$ evolves exactly the same way throughout the iterations of the   algorithm in~\cite{DGM:Arrow-Debreu} as it evolves throughout the iterations of the algorithm in~\cite{Duan-Mehlhorn:Arrow-Debreu-Market}. 
   
   As stated before, throughout the algorithm there are no \emph{type-3} agents. Hence, always we have $x_{max} = 1 + {1}/(R kn^2)$ where $k$ equals the number of \emph{type-1} agents in the equality network. After the equality network has changed $\mathcal{O}(n)$ times, the set $\Gamma(S)$ will have $\Omega(n)$ \emph{type-1} agents (Consequence of Observation \ref{only-type-1}) and then $\xmax = (1 + \Omega(1/n^3))$ in every iteration henceforth. These iterations must increase some prices further by a factor $\Omega(U^{\Omega(n)})$. This requires 
    $\Omega(\log_{(1+\Omega(\frac{1}{n^3}))}(U)^{\Omega(n)}) = \Omega(n^4\cdot\log(U))$ further iterations. 
\end{proof}
We next construct an instance $I'_n$ that separates the two algorithms in~\cite{Duan-Mehlhorn:Arrow-Debreu-Market,DGM:Arrow-Debreu}. 

\section{The instance separating the algorithms in~\cite{Duan-Mehlhorn:Arrow-Debreu-Market} and~\cite{DGM:Arrow-Debreu}}
 \label{instance_separation}
In our previous example we observe that the price of one of the goods ($g_1$) always increased in every iteration. The question is are there instances when
\begin{enumerate}
 \item There are $\Omega(n)$ goods with equilibrium prices in $\Omega(U^{\Omega(n)})$,
 \item Prices of each of the goods increase in different iterations.
\end{enumerate}
Such instances could make the algorithm in~\cite{Duan-Mehlhorn:Arrow-Debreu-Market} realize $\Omega(n^5\log(U))$ iterations. With this intention in mind we first prove and claim certain useful facts and observations.

\begin{fact}
  The number of primes less than or equal to an integer $x$ is $\Theta(\frac{x}{\ln(x)})$.
\end{fact}

\begin{lemma}
 \label{ratioboundDM}
  If $k_1$ and $k_2$ are prime integers and $k_1 \neq k_2$, and $\alpha_1 < k_1, \alpha_2 < k_2 $ are positive integers, then,
  \begin{itemize}
      \item $\frac{\alpha_1}{k_1} \neq \frac{\alpha_2}{k_2}$.
      \item $\frac{\alpha_1}{k_1}/\frac{\alpha_2}{k_2} \in [(1 + \frac{1}{k_1\cdot k_2})^{-1},(1 + \frac{1}{k_1\cdot k_2})]$.
  \end{itemize}
\end{lemma}
   
\begin{proof}
   Note that since both $k_1$ and $k_2$ are primes, the fraction ${k_1}/{k_2}$ cannot be reduced to lower terms.  Let ${\beta_1}/{\beta_2}$ be the reduced form of the fraction ${\alpha_1}/{\alpha_2}$. Now $\beta_1 < \alpha_1 < k_1$ and $\beta_2 < \alpha_2 < k_2$. Thus    ${\beta_1}/{\beta_2} \neq {k_1}/{k_2}$. Since reduced forms of the fractions are not equal, ${\alpha_1}/{\alpha_2} \neq {k_1}/{k_2}$.

   Since ${\alpha_1}/{k_1} \neq {\alpha_2}/{k_2}$, let us assume without loss of generality that ${\alpha_1}/{k_1} > {\alpha_2}/{k_2}$. In that case, ${\alpha_1}/{k_1} - {\alpha_2}/{k_2} = {(\alpha_1 k_2 - \alpha_2 k_1)}/{k_1 k_2} \geq {1}/{k_1 k_2} \geq {1}/{k_1 k_2} \cdot {\alpha_2}/{k_2}$ and hence ${\alpha_1}/{k_1} \geq (1 + {1}/{k_1 k_2}) \cdot {\alpha_2}/{k_2}$. Similarly if we assume ${\alpha_1}/{k_1} <{\alpha_2}/{k_2}$ we can conclude ${\alpha_2}/{k_2} \geq (1 + {1}/{k_1k_2})\cdot {\alpha_1}/{k_1})$. 
\end{proof}   
We will see ahead that these fractions would be the values of the surpluses of a set of agents and this ratio bound will help us identify the set of agents with high surplus (We refer to the definition of agents with high surplus used in~\cite{Duan-Mehlhorn:Arrow-Debreu-Market}).

\begin{observation}
 \label{technicalDM}
Let $p_i$ denote the budget of $b_i$ and price of $g_i$. Also $p_i = 1$ for all $i \in [n] \setminus [m]$ (for some $m<n$). Now, let us consider an equality Network $N_p$ with vertices $(s \cup B \cup G \cup t)$, where $B = \left\{b_1,b_2,...,b_n\right\}$ and $G = \left\{g_1,g_2,....,g_n\right\}$ with edges as follows,
\begin{itemize}
 \item $(s,b_i)$ with capacity $p_i \quad \forall i \in [n]$.
 \item $(g_i,t)$ with capacity $p_i\quad \forall i \in [n]$.
 \item $(b_i,g_j)$with capacity $ \infty \quad \forall i \in [n] \text{ and } j \in [m]$.
\end{itemize}
The surplus of the agent $b_i$ w.r.t. a balanced flow $f$ in $N_p$ is 
$r_f(b_i) = {(n-m)}/{n}$.
\end{observation}

\begin{proof}
 Note that the maximum flow is $\sum_{i \in [n]}^{}p_i$. Therefore for every balanced flow $f$, the sum of the surpluses of the agents $\sum_{i \in [n]}^{} r_f(b_i)$ will be equal to $\sum_{i \in [n]} p_i -\sum_{i \in [m]} p_i = n-m$. Clearly there is a flow that realizes this sum and has every surplus equal to $(n-m)/{n}$.  
\end{proof}
What is crucial in the above observation is that this claim is independent of the actual values of the prices of the goods in $\set{g_j}{j\in[m]}$. Before we construct our instance, we wish to give the reader an impression of how the prices rise in the same. Consider the utility matrix in Figure~\ref{Example 4}. So note that there are different green nodes (that correspond to different agents) with the same label. \emph{So an agent is uniquely identified by the unordered pair of the block he belongs to and the label of the node.}

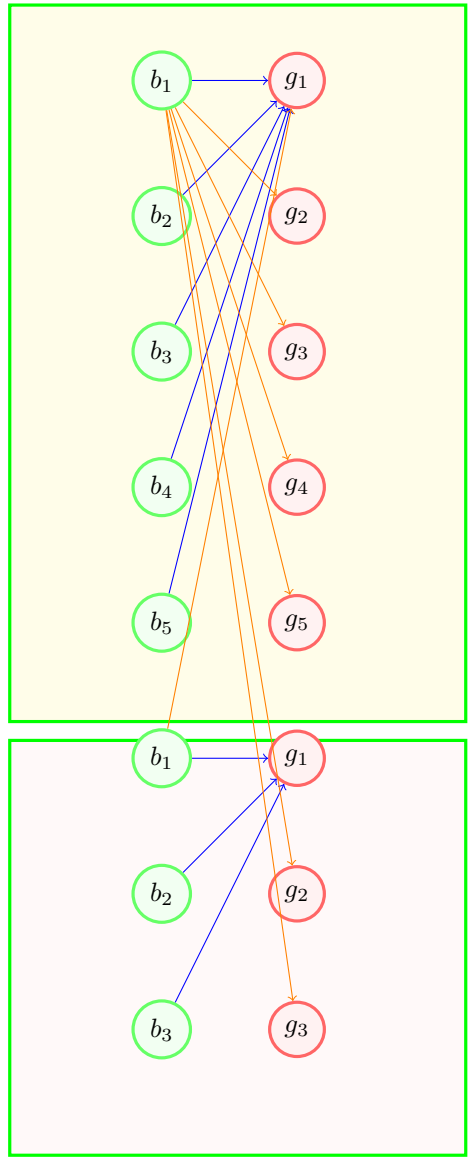
\begin{figure}[t!]
\centering
\begin{tikzpicture}
[
roundnode/.style={circle, draw=green!60, fill=green!5, very thick, minimum size=7mm},
roundnode2/.style={circle, draw=red!60, fill=red!5, very thick, minimum size=7mm},
]


\draw[green, fill=yellow!10, very thick] (-2,1) rectangle (4,-8.5) ;
\draw[green, fill= pink!10, very thick] (-2,-8.75) rectangle (4,-14.25);


\node[roundnode]      (l1)       {$b_1$};
\node[roundnode]      (l2)       [below=of  l1]{$b_2$};
\node[roundnode]      (l3)       [below=of  l2]{$b_3$};
\node[roundnode]      (l4)      [below=of  l3]{$b_4$};
\node[roundnode]      (l5)      [below=of l4]{$b_5$};

\node[roundnode2]      (r1)       [right=of l1]{$g_1$};
\node[roundnode2]      (r2)       [right=of l2] {$g_2$};
\node[roundnode2]      (r3)       [right=of l3] {$g_3$};
\node[roundnode2]      (r4)       [right=of l4] {$g_4$};
\node[roundnode2]      (r5)       [right=of l5]{$g_5$};


\node[roundnode]      (l'1)       [below=of  l5]{$b_1$};
\node[roundnode]      (l'2)       [below=of  l'1]{$b_2$};
\node[roundnode]      (l'3)       [below=of  l'2]{$b_3$};

\node[roundnode2]      (r'1)      [right=of l'1] {$g_1$};
\node[roundnode2]      (r'2)      [right=of  l'2]{$g_2$};
\node[roundnode2]      (r'3)      [right=of l'3]{$g_3$};


\draw[blue,->] (l1) -- (r1);
\draw[blue,->] (l2) -- (r1);
\draw[blue,->] (l3) -- (r1);
\draw[blue,->] (l4) -- (r1);
\draw[blue,->] (l5) -- (r1);

\draw[blue,->] (l'1) -- (r'1);
\draw[blue,->] (l'2) -- (r'1);
\draw[blue,->] (l'3) -- (r'1);

\draw[orange,->] (l1) -- (r2);
\draw[orange,->] (l1) -- (r3);
\draw[orange,->] (l1) -- (r4);
\draw[orange,->] (l1) -- (r5);
\draw[orange,->] (l1) -- (r'2);
\draw[orange,->] (l1) -- (r'3);

\draw[orange,->] (l'1)--(r1);

\end{tikzpicture}    
\caption{Utility graph where the agents are the nodes on the left and the goods are on the right. The agents have been partitioned into two disjoint blocks - $\mathit{block}_1$ (yellow block) and $\mathit{block}_2$ (pink block). In each block, agent $b_i$ owns good $g_i$. The blue lines from $b_i$ to $g_j$, indicates utility of $U$, and the orange ones a utility of $1$. }\label{Example 4}
\end{figure}
(See Figure 5) Initially the prices of all the goods are set to 1. The only edges present in the equality network are the blue edges of the utilty matrix shown in Figure 5. After computing the balanced flow, we can partition the agents into two blocks - $\mathit{block}_1$ and $\mathit{block}_2$, as we see in the figure itself. The agents in $\mathit{block}_1$ have a surplus of $\frac{4}{5}$ and the ones in $\mathit{block}_2$ have $\frac{2}{3}$. Thus ratio of their surpluses equals $\frac{12}{10} = 1 + \frac{1}{5} > 1 + \frac{1}{n}$, where $n$ here is $8$.Thus during the first iteration $S = \mathit{block}_1$ ( we refer to the $S$ in~\cite{Duan-Mehlhorn:Arrow-Debreu-Market}). We increase the price of good $g_1$ in $\mathit{block}_1$ by $1 + \frac{1}{Rn^3}$. However from  Observation 3.3, we can conclude that the surplus vector still remains the same for all nodes in $\mathit{block}_1$ ( = $\frac{4}{5}$ for each). Thus the algorithm in~\cite{Duan-Mehlhorn:Arrow-Debreu-Market} increases the price of good $g_1$ in $\mathit{block}_1$ until the equality network changes (when its price = $U$). Thereafter agent $b_1$ of $\mathit{block}_1$ will be interested to invest in all the goods except $g_1$ of $\mathit{block}_2$ . In the new equality network, the surplus of all agents in $\mathit{block}_1$ is zero and $\mathit{block}_2$ is the new set of agents with high surplus. Like earlier, the surplus of the agents in $\mathit{block}_2$ does not change until, the equality network changes (when the price of good $g_1$ in $\mathit{block}_2$ is $U^2$). The important note is that the price rise in both the blocks occur in disjoint iterations of the algorithm.
\begin{figure}[t!]
 
 \begin{subfigure}{1.0\textwidth}
   \begin{tikzpicture}
[
roundnode/.style={circle, draw=green!60, fill=green!5, very thick, minimum size=3mm},
roundnode3/.style={circle, draw=green!60, fill=green!5, very thick, minimum size=3mm},
roundnode2/.style={circle, draw=red!60, fill=red!5, very thick, minimum size=3mm},
roundnode23/.style={circle, draw=red!60, fill=red!5, very thick, minimum size=3mm},
roundnode4/.style={circle, draw=red!60, fill=red!55, very thick, minimum size=5mm},
roundnode6/.style={circle, draw=red!60, fill=red!55, very thick, minimum size=5mm},
whitenodes/.style={circle, draw=white!60, fill=white!55, very thick, minimum size=5mm},
lineo/.style={orange,->},
lineb/.style={blue,->,},
tlineb/.style={blue,->,ultra thick},
tlineo/.style={orange,->,ultra thick},
]




\node[whitenodes]     (w) {};
\node[roundnode3]      (b1)       [right=160pt of w] {$b_1$};
\node[roundnode3]      (b2)       [below=3pt of  b1]{$b_2$};
\node[roundnode3]      (b3)       [below=3pt of  b2]{$b_3$};
\node[roundnode3]      (b4)      [below=3pt of  b3]{$b_4$};

\node[roundnode23]      (g1)       [right=of b1]{$g_1$};
\node[roundnode23]      (g2)       [right=of b2] {$g_2$};
\node[roundnode23]      (g3)       [right=of b3] {$g_3$};
\node[roundnode23]      (g4)      [right=of b4] {$g_4$};


\node[roundnode3]      (b'1)       [below=3pt of  b4]{$b_1$};
\node[roundnode3]      (b'2)       [below=3pt of  b'1]{$b_2$};
\node[roundnode3]      (b'3)       [below=3pt of  b'2]{$b_3$};

\node[roundnode23]      (g'1)      [right=of b'1] {$g_1$};
\node[roundnode23]      (g'2)       [right=of  b'2]{$g_2$};
\node[roundnode23]      (g'3)     [right=of b'3]{$g_3$};


\draw[tlineb] (b1) -- (g1);
\draw[tlineb] (b2) -- (g1);
\draw[tlineb] (b3) -- (g1);
\draw[tlineb] (b4) -- (g1);

\draw[tlineb] (b'1) -- (g'1);
\draw[tlineb] (b'2) -- (g'1);
\draw[tlineb] (b'3) -- (g'1);

\node[whitenodes]   (s3) [left=6pt of b1] {$\frac{3}{4}$};
\node[whitenodes]   (s2) [left=6pt of b2] {$\frac{3}{4}$};
\node[whitenodes]   (s1) [left=6pt of b3] {$\frac{3}{4}$};
\node[whitenodes]   (s2') [left=6pt of b4] {$\frac{3}{4}$};

\node[whitenodes]   (s'2)  [left=6pt of b'1] {$\frac{2}{3}$};
\node[whitenodes]   (s'1)  [left=6pt of b'2] {$\frac{2}{3}$};
\node[whitenodes]   (s'2') [left=6pt of b'3] {$\frac{2}{3}$};

\node[whitenodes]   (s3)   [right=6pt of g1] {$1$};
\node[whitenodes]   (s2)   [right=6pt of g2] {$1$};
\node[whitenodes]   (s1)   [right=6pt of g3] {$1$};
\node[whitenodes]   (s2')   [right=6pt of g4] {$1$};

\node[whitenodes]   (s2)   [right=6pt of g'1] {$1$};
\node[whitenodes]   (s1)   [right=6pt of g'2] {$1$};
\node[whitenodes]   (s2')   [right=6pt of g'3] {$1$};

\node[roundnode4]  (g1)   [right=of b1] {$g_1$};

\end{tikzpicture}
   \caption{During the first iteration of the algorithm in~\cite{Duan-Mehlhorn:Arrow-Debreu-Market} on $I'_n$. We have $S = \mathit{block}_1$.} \label{fig:4a}
 \end{subfigure}\smallskip
 
 \begin{subfigure}{1.0\textwidth}
   \begin{tikzpicture}
[
roundnode/.style={circle, draw=green!60, fill=green!5, very thick, minimum size=3mm},
roundnode3/.style={circle, draw=green!60, fill=green!5, very thick, minimum size=3mm},
roundnode2/.style={circle, draw=red!60, fill=red!5, very thick, minimum size=3mm},
roundnode23/.style={circle, draw=red!60, fill=red!5, very thick, minimum size=3mm},
roundnode4/.style={circle, draw=red!60, fill=red!55, very thick, minimum size=5mm},
roundnode6/.style={circle, draw=red!60, fill=red!55, very thick, minimum size=5mm},
whitenodes/.style={circle, draw=white!60, fill=white!55, very thick, minimum size=5mm},
lineo/.style={orange,->},
lineb/.style={blue,->,},
tlineb/.style={blue,->,ultra thick},
tlineo/.style={orange,->,ultra thick},
]

\node[whitenodes]     (w) {};
\node[roundnode3]      (b1)       [right=160pt of w] {$b_1$};
\node[roundnode3]      (b2)       [below=3pt of  b1]{$b_2$};
\node[roundnode3]      (b3)       [below=3pt of  b2]{$b_3$};
\node[roundnode3]      (b4)      [below=3pt of  b3]{$b_4$};

\node[roundnode23]      (g1)       [right=of b1]{$g_1$};
\node[roundnode23]      (g2)       [right=of b2] {$g_2$};
\node[roundnode23]      (g3)       [right=of b3] {$g_3$};
\node[roundnode23]      (g4)      [right=of b4] {$g_4$};


\node[roundnode3]      (b'1)       [below=3pt of  b4]{$b_1$};
\node[roundnode3]      (b'2)       [below=3pt of  b'1]{$b_2$};
\node[roundnode3]      (b'3)       [below=3pt of  b'2]{$b_3$};

\node[roundnode23]      (g'1)      [right=of b'1] {$g_1$};
\node[roundnode23]      (g'2)       [right=of  b'2]{$g_2$};
\node[roundnode23]      (g'3)     [right=of b'3]{$g_3$};


\draw[tlineb] (b1) -- (g1);
\draw[tlineb] (b2) -- (g1);
\draw[tlineb] (b3) -- (g1);
\draw[tlineb] (b4) -- (g1);

\draw[tlineb] (b'1) -- (g'1);
\draw[tlineb] (b'2) -- (g'1);
\draw[tlineb] (b'3) -- (g'1);

\draw[tlineo] (b1)--(g2);
\draw[tlineo] (b1)--(g3);
\draw[tlineo] (b1)--(g4);
\draw[tlineo] (b1)--(g'2);
\draw[tlineo] (b1)--(g'3);

\node[whitenodes]   (s3) [left=6pt of b1] {$0$};
\node[whitenodes]   (s2) [left=6pt of b2] {$0$};
\node[whitenodes]   (s1) [left=6pt of b3] {$0$};
\node[whitenodes]   (s2') [left=6pt of b4] {$0$};

\node[whitenodes]   (s'2)  [left=6pt of b'1] {$\frac{2}{3}$};
\node[whitenodes]   (s'1)  [left=6pt of b'2] {$\frac{2}{3}$};
\node[whitenodes]   (s'2') [left=6pt of b'3] {$\frac{2}{3}$};

\node[whitenodes]   (s3)   [right=6pt of g1] {$U$};
\node[whitenodes]   (s2)   [right=6pt of g2] {$1$};
\node[whitenodes]   (s1)   [right=6pt of g3] {$1$};
\node[whitenodes]   (s2')   [right=6pt of g4] {$1$};

\node[whitenodes]   (s2)   [right=6pt of g'1] {$1$};
\node[whitenodes]   (s1)   [right=6pt of g'2] {$1$};
\node[whitenodes]   (s2')   [right=6pt of g'3] {$1$};

\node[roundnode4]  (g1)   [right=of b'1] {$g_1$};

\end{tikzpicture}
   \caption{The equality network changes for the first time. Orange edges appear. $S = \mathit{block}_2$.} \label{fig:4b}
 \end{subfigure}\smallskip
 
 \begin{subfigure}{1.0\textwidth}
   \begin{tikzpicture}
[
roundnode/.style={circle, draw=green!60, fill=green!5, very thick, minimum size=3mm},
roundnode3/.style={circle, draw=green!60, fill=green!5, very thick, minimum size=3mm},
roundnode2/.style={circle, draw=red!60, fill=red!5, very thick, minimum size=3mm},
roundnode23/.style={circle, draw=red!60, fill=red!5, very thick, minimum size=3mm},
roundnode4/.style={circle, draw=red!60, fill=red!55, very thick, minimum size=5mm},
roundnode6/.style={circle, draw=red!60, fill=red!55, very thick, minimum size=5mm},
whitenodes/.style={circle, draw=white!60, fill=white!55, very thick, minimum size=5mm},
lineo/.style={orange,->},
lineb/.style={blue,->,},
tlineb/.style={blue,->,ultra thick},
tlineo/.style={orange,->,ultra thick},
]
\node[whitenodes]     (w) {};
\node[roundnode3]      (b1)       [right=160pt of w] {$b_1$};
\node[roundnode3]      (b2)       [below=3pt of  b1]{$b_2$};
\node[roundnode3]      (b3)       [below=3pt of  b2]{$b_3$};
\node[roundnode3]      (b4)      [below=3pt of  b3]{$b_4$};

\node[roundnode23]      (g1)       [right=of b1]{$g_1$};
\node[roundnode23]      (g2)       [right=of b2] {$g_2$};
\node[roundnode23]      (g3)       [right=of b3] {$g_3$};
\node[roundnode23]      (g4)      [right=of b4] {$g_4$};


\node[roundnode3]      (b'1)       [below=3pt of  b4]{$b_1$};
\node[roundnode3]      (b'2)       [below=3pt of  b'1]{$b_2$};
\node[roundnode3]      (b'3)       [below=3pt of  b'2]{$b_3$};

\node[roundnode23]      (g'1)      [right=of b'1] {$g_1$};
\node[roundnode23]      (g'2)       [right=of  b'2]{$g_2$};
\node[roundnode23]      (g'3)     [right=of b'3]{$g_3$};


\draw[tlineb] (b1) -- (g1);
\draw[tlineb] (b2) -- (g1);
\draw[tlineb] (b3) -- (g1);
\draw[tlineb] (b4) -- (g1);

\draw[tlineb] (b'1) -- (g'1);
\draw[tlineb] (b'2) -- (g'1);
\draw[tlineb] (b'3) -- (g'1);

\draw[tlineo] (b1)--(g2);
\draw[tlineo] (b1)--(g3);
\draw[tlineo] (b1)--(g4);
\draw[tlineo] (b1)--(g'2);
\draw[tlineo] (b1)--(g'3);
\draw[tlineo] (b'1)--(g1);

\node[whitenodes]   (s3) [left=6pt of b1] {$0$};
\node[whitenodes]   (s2) [left=6pt of b2] {$0$};
\node[whitenodes]   (s1) [left=6pt of b3] {$0$};
\node[whitenodes]   (s2') [left=6pt of b4] {$0$};

\node[whitenodes]   (s'2)  [left=6pt of b'1] {$0$};
\node[whitenodes]   (s'1)  [left=6pt of b'2] {$0$};
\node[whitenodes]   (s'2') [left=6pt of b'3] {$0$};

\node[whitenodes]   (s3)   [right=6pt of g1] {$U$};
\node[whitenodes]   (s2)   [right=6pt of g2] {$1$};
\node[whitenodes]   (s1)   [right=6pt of g3] {$1$};
\node[whitenodes]   (s2')   [right=6pt of g4] {$1$};

\node[whitenodes]   (s2)   [right=6pt of g'1] {$U^2$};
\node[whitenodes]   (s1)   [right=6pt of g'2] {$1$};
\node[whitenodes]   (s2')   [right=6pt of g'3] {$1$};

\end{tikzpicture}
   \caption{The equality network changes for the second time. The edge from $(b_1,\mathit{block}_2)$ to $(g_1,\mathit{block}_1)$ appears.} \label{fig:4c}
 \end{subfigure}
 \caption{Running the algorithm in~\cite{Duan-Mehlhorn:Arrow-Debreu-Market} on the instance similar to that in Figure \ref{Example 4} . The numbers on left represent surplus of the agents and that on the right the prices of the goods.}
\end{figure}
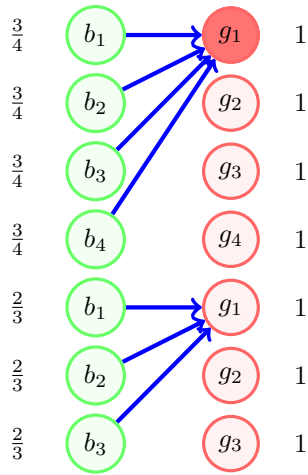
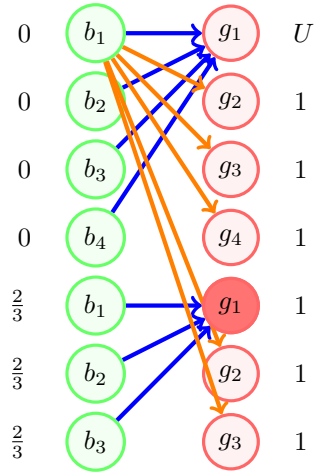
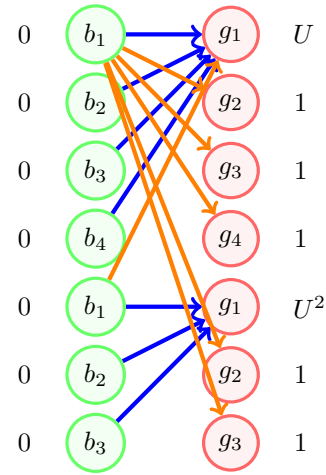
The key idea is that we can create instances where the agents can be partitioned into disjoint block based on their surpluses. At any iteration of the algorithm, the ratio of the surpluses between any two blocks would be greater than $(1+\frac{1}{n})$. Or in other words only one block is present in $S$ at every iteration. This way we can ensure that prices of goods in different blocks increase in different iterations. Thus if we can partition the agents into $k$ blocks (prices of goods in each block being updated in disjoint iterations) we can ensure that the algorithm in~\cite{Duan-Mehlhorn:Arrow-Debreu-Market} incurs $\Omega( \sum_{i \in [k]} \log_{1+\frac{1}{R \cdot n^3}}(U)^i)) = \Omega (n^3\log(U) \cdot (\sum_{i \in[k]} i) = \Omega(n^3k^2\log(U))$.
 
 \paragraph*{Constructing the Instance $I'_n$: Partitioning the Agents into Blocks and Defining the Utility Matrix:}
 
 \begin{enumerate}
  
  \item Initially set every $u(b_i,g_j) = 0$ for every agent $b_i$ and good $g_j$.
  
  \item For every prime $i < n^{\frac{1}{3}}$ and $j<i$, we have a $\mathit{block}_{i,j}$, comprising $i$ agents $(b_1,\mathit{block}_{i,j})$ to $(b_i,\mathit{block}_{i,j})$ and $i$ goods $(g_1,\mathit{block}_{i,j})$ to $(g_i,\mathit{block}_{i,j})$. Agent $(b_k,\mathit{block}_{ij})$ owns $(g_k, \block_{ij})$ for all $k \in [i]$.  
   
  \item Let $u((b_l, \block_{i,j}),(g_{l'},\block_{i',j'}))$ denote the utility of agent $(b_l, \block_{i,j})$ from $(g_{l'}, \block_{i',j'})$.  For every $\mathit{block}_{i,j}$, we set  $u((b_l,\mathit{block}_{i,j}),(g_m,\mathit{block}_{i,j})) = 2U $ for all $l\in [i] $ and all $ m \in [j]$. 
  
   \item For every block $\mathit{block}_{i,j}$, we call the goods $(g_m,\mathit{block}_{i,j})$ \emph{high-demand goods} of $\mathit{block}_{i,j}$ if and only if $m \in [j]$, (as all agents in $\mathit{block}_{i,j}$ are interested in them) and the remaining goods as \emph{low-demand goods} of $\mathit{block}_{i,j}$.
  
  \item Let $\pi$ be the sequence of blocks sorted in decreasing order of $\frac{i-j}{i}$. Let $\pi_i$ denotes the $i^{\mathit{th}}$ block in the sequence $\pi$, and $|\pi|$ denotes the number of length of the sequence which equals the number of blocks $\in \Theta(\frac{n^{\frac{2}{3}}}{\log^2(n)})$ (Follows from Fact 3.1).
  
  \item For all $j \in [\lvert \pi \rvert]$ we set $u((b_l,\pi_1),(g_m,\pi_j)) = 2$, for every $l$ and $m$ such that $(g_l,\pi_1)$ is a \emph{high-demand good} in $\pi_1$, and $(g_m,\pi_j)$ is a  \emph{low-demand good} in $\pi_j$.
  
  \item For all $i\in [|\pi|-1]$, we set $u((b_l,\pi_{i+1}),(g_m,\pi_{i})) = 2$ for every $l$ and $m$ such that  $(g_l,\pi_{i+1}) $ is a \emph{high-demand good} in $\pi_i$, and $(g_m,\pi_{i}) $ is a \emph{high demand-good} in $\pi_{i}$.
  
  \item Finally we define the last block $\block_{n,n}$ comprising agents $(b_1,\block_{n,n})$ to $(b_n,\block_{n,n})$ and goods $(g_1,\block_{n,n})$ to $(g_n,\block_{n,n})$ and we define a few more entries in the utility matrix as follows,
  \begin{itemize}
      \item $u((b_i,\block_{n,n}),(g_i,\block_{n,n})) = U \quad \forall i \in [n]$.
      \item $u((b_i,\block_{n,n}), (g_1,\pi_1)) = 1 \quad \forall i \in [n]$.
      \item $u((b_l,\pi_1),(g_i,\block_{n,n})) = 1 \quad \forall i \in [n]$, for every $l$ such that  $(g_l,\pi_1)$ is a \emph{high demand good} in $\pi_1$.
  \end{itemize}
    
 \end{enumerate}
 
 \begin{observation}
  The number of agents and goods in $I'_n$  $\in \Omega(n)$.
 \end{observation}
  
  \begin{proof}
     The number of agents in $\block_{n,n}$ is $n$.  
  \end{proof}
It can be verified that the instance $I'_n$ is irreducible. On the very first iteration of the algorithm in~\cite{Duan-Mehlhorn:Arrow-Debreu-Market} (when prices of all goods are set to unity), we observe that the equality network is disconnected. All the agents in $\block_{n,n}$ have zero surplus as they invest completely in the good they own. After computing the balanced flow, if we partition the agents based on their surpluses, we have that all agents from $\mathit{block}_{i,j}$ have surplus ${(i-j)}/{i}$. Note that the ratio of the surpluses of the agents in any two disjoint blocks is strictly larger than $ 1 + {1}/{n^{\frac{2}{3}}}$ - To see this consider any two arbitrary blocks $\mathit{block}_{i,j}$ and $\mathit{block}_{i',j'}$. If $i \neq i'$ then the ratio of the surpluses is larger than $1 + {1}/{n^{\frac{2}{3}}}$ (follows from the fact that $i<n^{\frac{1}{3}}$ and Lemma \ref{ratioboundDM}). If $i =i'$, assume without loss of generality that $j' > j$,  then the ratio of the surpluses is $(i-j)/(i-j') = 1+ (j'-j)/(i-j') \geq  1 + 1/i > 1 + {1}/{n^{\frac{2}{3}}}$ (Since $i < 1/n^{\frac{1}{3}}$).   Thus,only the agents of $\pi_1$ belong in $S$ after the first iteration. The algorithm steadily increases the prices of the \emph{high-demand goods} in $\pi_1$, but the surplus vector remains unchanged (follows from Observation \ref{technicalDM}). Thus only agents of $\pi_1$ belong in $S$ for $n^3 \cdot log(U)$ iterations, when the equality network finally changes. After the change notice that the surplus of all the agents in $\pi_1$ becomes zero and only agents of $\pi_2$belong in $S$.
 
 Inductively one can make the same argument, during execution of the algorithm, when only agents of $\pi_i$ belong in $S$ and the surpluses of all the agents in the blocks $\pi_j$ with $ j<i$ and in $\block_{n,n}$ are zero and prices of the \emph{high demand goods} in each of the blocks $\pi_j$ with $j < i$ is $U^{j}$. The price increase of the \emph{high-demand goods} in $\pi_i$ continues for $\Omega(n^3 \cdot \log((U)^{i}))$ many iterations, until the equality network changes or equivalently when the equality edges from $(b_l,\pi_i)$ to $(g_m,\pi_{i-1})$ arise, for every $l$ and $m$ such that $(g_l,\pi_i)$ is a \emph{high-demand good} in $\pi_i$ and $(g_m,\pi_{i-1})$ is a \emph{high-demand good} in $\pi_{i-1}$. With the new equality edges appearing, new augmenting paths also appear in the equality network that reduce the surplus of all agents in $\pi_i$ to zero. However since there are no edges connecting the agents in blocks $\pi_1,\pi_2,...,\pi_i$ to the \emph{high-demand goods} of blocks $\pi_{i+1},...\pi_{|\pi|}$, (and therefore no augmenting path via the agents in the blocks $\pi_{i+1},...\pi_{|\pi|}$ to reduce surpluses of the agents in blocks $\pi_j$ with $j>i$) the surpluses of the agents of the latter blocks remain unchanged. Henceforth only agents $\pi_{i+1}$ constitute $S$.

\begin{theorem}
 \label{DMisworse}
 The number of iterations executed by the algorithm in~\cite{Duan-Mehlhorn:Arrow-Debreu-Market} on the instance $I'_n$ is $\Omega((n^{4+\frac{1}{3}}\log(U))/\log(n)^4)$.
\end{theorem}
\begin{proof}
With the utility matrix defined as above, the execution of the algorithm in~\cite{Duan-Mehlhorn:Arrow-Debreu-Market} incurs 
\[     \Omega(\sum \limits_{i \in [\Theta(\frac{n^{\frac{2}{3}}}{\log^2(n)})]}^{} n^3 \cdot \log(U^i)) =  \Omega\left(\frac{n^{4+\frac{1}{3}} \cdot \log(U)}{\log^4 n}\right)\]  iterations.
\end{proof}

\bibliographystyle{plain}








\end{document}